%% file: main.tex
\documentclass{amsart}[a4paper]

\makeatletter
\let\mymakefnmark\@makefnmark
\let\mythefnmark\@thefnmark

\newcommand{\restorefn}{\let\@makefnmark\mymakefnmark
\let\mythfnmakr\@thefnmark}
\makeatother

\usepackage{etoolbox}

\usepackage{graphicx}
\usepackage{pdflscape}
\usepackage{graphicx}%
\usepackage{multirow}%
\usepackage{amsmath,amssymb,amsfonts}%
\usepackage[foot]{amsaddr}
\usepackage{mathrsfs}%
\usepackage[title]{appendix}%
\usepackage{xcolor}%
\usepackage{textcomp}%
\usepackage{manyfoot}%
\usepackage{booktabs}%
\usepackage{algorithm}%
\usepackage{algorithmicx}%
\usepackage{algpseudocode}%
\usepackage{listings}%
\usepackage{xspace}
\usepackage{hyperref}
\usepackage{subcaption}
\usepackage{url}
\usepackage[disable]{todonotes}
\usepackage{wrapfig,booktabs}

\allowdisplaybreaks

\def\BibTeX{{\rm B\kern-.05em{\sc i\kern-.025em b}\kern-.08em
    T\kern-.1667em\lower.7ex\hbox{E}\kern-.125emX}}

\newtheorem{theorem}{Theorem}
\newtheorem{proposition}[theorem]{Proposition}%
\newtheorem{theorem*}{Theorem}

\newcommand{\kgrodel}{$k$-\textsc{GRoDel}\xspace}
\newcommand{\fd}[3]{d^{f}_{#1}(#2,#3)}
\newcommand{\tool}[1]{\textsc{#1}}
\newcommand{\THR}{\ensuremath{R_{h}}\xspace}

\newcommand{\egc}{e.\,g.,\xspace}
\newcommand{\iec}{i.\,e.,\xspace}
\newcommand{\etal}{et al.\xspace}

\newcommand{\mat}[1]{\ensuremath{\mathbf{#1}}}   
\newcommand{\vect}[1]{\mathbf{#1}}   
\newcommand{\Lap}{\mat{L}}

\newcommand{\GreedyTHR}{\texttt{GreedyTHR}\xspace}
\newcommand{\GreedyFI}{\texttt{GreedyFI}\xspace}

\newcommand{\Lpinv}{\ensuremath{\mat{L}^\dagger}}

\newcommand{\effres}[3]{\mat{r}_{#1}(#2,#3)}
\newcommand{\fstar}{u^*}
\newcommand{\trace}[1]{\operatorname{tr}(#1)}
\newcommand{\easubeb}{(\vect{e}_a - \vect{e}_b)}
\newcommand{\norm}[1]{\left\lVert#1\right\rVert}
\newcommand{\uvec}[1]{\vect{e_{#1}}}
\newcommand{\loss}[2]{\text{loss}(#1,#2)}
\newcommand{\argmax}{\operatorname{argmax}}
\newcommand{\Oh}{\ensuremath{\mathcal{O}}}

\newcommand{\eqncontinue}{~~~~}
\newcommand{\eqnlinebreak}{\\&\eqncontinue }

\usepackage{fancyhdr}

\renewcommand{\thefootnote}{\fnsymbol{footnote}}

\begin{document}

\bibliographystyle{splncs04}

\title{Introducing Total Harmonic Resistance for \\ Graph Robustness under Edge Deletions\restorefn\footnotemark}

\author{Lukas Berner}
\email{lukas.berner@hu-berlin.de}

\author{Henning Meyerhenke}
\email{meyerhenke@hu-berlin.de}
\address[Lukas Berner, Henning Meyerhenke]{Department of Computer Science, Humboldt-Universität zu Berlin, Unter den Linden 6, 10099 Berlin, Germany}

\begin{abstract}
Assessing and improving the robustness of a graph $G$ are critical steps in network design and analysis. 
To this end, we consider the optimisation problem of removing $k$ edges from $G$ such that the resulting graph has 
minimal robustness, simulating attacks or failures.

In this paper, we propose total harmonic resistance as a new robustness measure for this purpose
-- and compare it to the recently proposed forest index [Zhu et al., IEEE Trans.\ Inf.\ Forensics and Security, 2023]. 
Both measures are related to the established total effective resistance measure, but their advantage
is that they can handle disconnected graphs. This is also important for originally connected graphs
due to the removal of the $k$ edges. 
To compare our measure with the forest index, we first investigate exact solutions for small examples.
The best $k$ edges to select when optimizing for the forest index lie at the periphery. 
Our proposed measure, in turn, prioritizes more central edges, which should be beneficial for most applications. 
Furthermore, we adapt a generic greedy algorithm to our optimization problem with the total harmonic
resistance. With this algorithm, we perform a case study on the Berlin road network 
and also apply the algorithm to established benchmark graphs. 
The results are similar as for the small example graphs above and indicate the higher suitability of the new measure.

\textbf{Keywords:} Graph robustness optimization, infrastructure protection, total harmonic resistance, forest index, effective resistance
\end{abstract}

\maketitle

\stepcounter{footnote}\footnotetext{\scshape This preprint has
not undergone peer review (when applicable) or any post-submission improvements or
corrections. The Version of Record of this contribution will be published at ECMLPKDD 2024 soon.}
\setcounter{footnote}{0}

\renewcommand{\thefootnote}{\arabic{footnote}}

\section{Introduction}
\label{sec:intro}
The analysis of network\footnote{We use the terms \textit{network} and \textit{graph} interchangeably.}
topologies, a major subarea of data science on network data,
is key to understanding the functionality, dynamics, and
evolution of networks~\cite{barabasi,Newman2018networks}. An important property of a network in this context is its \emph{robustness},
\iec its ability to withstand failures of its components
(or the extent of this ability)~\cite{barabasi}.
As an example, a typical question is whether a network remains (mostly) connected
if a certain fraction of its vertices and/or edges are deleted~\cite[Ch.~15]{Newman2018networks}.
Despite the widespread use of vertex deletions, edge deletions can be more appropriate depending
on the modeled phenomenon. Examples include, among others, road blocks in street or public transportation
networks; pollution in water distribution networks; disruption
of gas pipelines, energy grids, or computer/telecommunication networks.
Such deletions may occur as a result of failure or of an attack; robustness thus is a
critical design issue that arises in many application areas~\cite{freitas2022graph}, \egc various
public infrastructures~\cite{oded,yakup,jose,water}.

Due to economic reasons, it is unrealistic that all network components can be protected
with the same effort. Thus, with the protection of critical infrastructure 
as application in mind, we consider the following optimization problem: given a graph $G = (V, E)$ and a
budget of $k$ graph edges to be removed, find the subset $S \subset E$ such that
the robustness of $G' = (V, E \setminus S)$ is minimized. This problem, which we call
\kgrodel (short for \emph{graph robustness problem with $k$ deletions}), models a concurrent
attack (or failure). The solution indicates which set of edges should be particularly safeguarded,
\egc segments in a water distribution network.
Clearly, for a particular application, one must also instantiate this generic problem with
a sensible notion of robustness.

Not surprisingly, numerous robustness measures have been proposed in the literature~\cite{barabasi,jose}.
For the related problem of optimizing the robustness by \emph{adding} $k$ edges (called $k$-GRIP in Ref.~\cite{PredariBKM23}, short for \emph{graph robustness improvement problem}),
\emph{total effective resistance} was established as a meaningful robustness measure in various
scenarios~\cite{Ellens2011,ps18,Wang2014ImprovingRO,DBLP:conf/asunam/PredariKM22}.
Effective resistance is a pairwise metric on the vertices of $G$;
intuitively, it becomes small if there are many short paths
between two vertices. Two disconnected vertices have infinite effective
resistance, though. When total graph resistance were used in \kgrodel, a trivial solution to maximize it 
would thus be to disconnect $G$.
Yet, from an application's point of view,
disconnecting a small part from the vast majority of the graph may be less problematic than 
a bottleneck (or a disconnection) between two large parts.
Liu \etal~\cite{liu24fast} handles this issue by demanding that $G$ is still connected after edge removal.
Given an infrastructure scenario, this is a rather unnatural assumption.
Zhu \etal~\cite{zhu2023measures} address the issue by proposing the forest index, $R_f(G)$, as robustness
measure. Instead of effective resistance, $R_f(G)$ sums up the closely related
forest distance~\cite{chebotarev2000forest} for all vertex pairs. 
Forest distance is derived from the number of certain rooted forests in $G$.
It yields finite distance values also for disconnected vertex pairs.

\paragraph{Contribution.}
We show in this paper that \kgrodel using the forest index favors peripheral edges in many networks.
We deem this behavior unintuitive and, most importantly, undesirable for most applications. That is
why we propose \emph{total harmonic resistance} $\THR(\cdot)$ instead. This measure adds up the 
reciprocal of effective resistance for all vertex pairs, leading to a zero contribution of disconnected 
pairs (details in Section~\ref{sec:prelim}).
The use of $\THR(\cdot)$ may seem like a straightforward extension to handle deletions given that
the use of reciprocals is known for the popular (harmonic) closeness centrality (based on the ordinary graph distance)
to handle disconnectedness~\cite{Newman2018networks}. Nonetheless, we are to our knowledge
the first to investigate this notion of robustness (also see Section~\ref{sec:rel-work}).

To substantiate the higher suitability of $\THR(\cdot)$ compared to the forest index in \kgrodel,
we first examine optimal solutions for small graphs with the two measures 
(Section~\ref{sec:exact-solutions}). They clearly show that $\THR(\cdot)$ favors more central edges
than the forest index and also finds balanced cuts in examples with suitable $k$.

Since exact solutions for either \kgrodel measure are expensive to compute, we adapt 
in Section~\ref{sec:heuristic} the general greedy algorithm used in several previous papers for related problems.
For $\THR(\cdot)$, our greedy algorithm differentiates between bridge edges and other edges.
When a bridge edge is removed, a simple update operation for the Laplacian pseudoinverse does not work.
For these cases, we thus provide specialized update functions.
For the forest index, we derive a connection to total effective resistance, which allows the re-use of optimized Laplacian pseudoinverse solvers instead of more general matrix inversion solvers.

Our experiments (Sec.~\ref{sec:experiments}) include a case study on the road network of (a part of)
Berlin, Germany, as well as numerous public benchmark graphs used before in related work. They show:
(i) visually, the case study results indicate that the greedy solution for $\THR(\cdot)$ prefers more central edges than the one with the forest index.
Maybe not surprisingly, one can find an even better solution regarding $\THR(\cdot)$ by choosing natural cut edges (river bridges in the road network) manually, which underlines the expressiveness of the new measure;
(ii) for the benchmark graphs, a ranking based on closeness centrality confirms that greedy solutions
of $\THR(\cdot)$ lead to more central edges than the forest index in most cases, too.

\section{Problem Statement and a new Robustness Measure}
\label{sec:prelim}
\subsection{Problem Statement and Notation}
The input to \kgrodel is an integer $k\in\mathbb{N}$ and a simple undirected graph $G=(V,E)$ 
with $|V| = n$, $|E|=m$. Given $S\subset E$, let $G'=(V, E\setminus S)$ be the graph with the edges from $S$ removed.
\kgrodel aims at finding $S$ with $|S|=k$ such that the robustness of $G'$ is minimized 
(\iec $\THR(G')$ is minimized or $R_f(G')$ is maximized, respectively).

We use well-known matrix representations of graphs.
$\Lap_G = \mat{D}-\mat{A} \in \mathbb{R}^{n\times n}$ is the Laplacian matrix of $G$, where $\mat{D}$ is the vertex degree matrix and $\mat{A}$ is the adjacency matrix.
$\Lap_G$ is symmetric (since $G$ is undirected) and has zero row and column sums ($\Lap\mat{1}={0} = \mat{1}^T\Lap$).
Since $\Lap_G$ is not invertible, the Moore-Penrose pseudoinverse \Lpinv\ is used instead (cf.~\cite{DBLP:journals/socnet/BozzoF13}).
When $G$ has multiple components, $\Lap_G$ is a (permuted) block diagonal matrix where each block corresponds to one component of $G$.

\subsection{Robustness Measures}

\paragraph*{Effective Resistance.}
Viewing the graph as an electrical circuit where each edge is a resistor, the effective resistance is the
potential difference between two nodes $u$ and $v$ when injecting [extracting] a unit current at $u$ [$v$]~\cite{Klein93}.
It can be computed via $\Lpinv$: $\effres{G}{u}{v} = \Lpinv_G[u,u] - 2 \Lpinv_G[u,v] + \Lpinv_G[v,v]$ for nodes in the same component of $G$. For disconnected pairs, the resistance is infinite.
As a robustness measure, one can take the sum over all pairwise effective resistances to compute the total effective resistance $R_r(G) = \sum_{u < v} \effres{G}{u}{v}$, which has previously been used as optimization target for $k$-GRIP~\cite{Wang2014ImprovingRO,ps18,DBLP:conf/asunam/PredariKM22} in graphs with only one component.
Combining both equations above results in a simple trace-based formula~\cite{DBLP:journals/socnet/BozzoF13}:
\begin{equation}\label{eqn:effective_resistance_trace}
    R_r(G) = n \cdot \trace{\Lpinv_G}.
\end{equation}

\paragraph*{Forest Index.}
To address the issue of disconnected graphs, other robustness measures are required.
The \emph{forest index}, based on the \emph{forest distance}~\cite{chebotarev2000forest} $\fd{G}{\cdot}{\cdot}$, was proposed by Zhu et al.~\cite{zhu2023measures}.
Similar to effective resistance, the forest distance is based on the \emph{forest matrix} $\Omega = (L + I)^{-1}$, with $\fd{G}{u}{v} = \Omega[u,u] - 2 \Omega[u,v] + \Omega[v,v]$.
The forest distance is closely related to effective resistance (for details see Section~\ref{sec:heuristic}), but yields finite values also for disconnected vertex pairs. Similar to total effective resistance, the forest index is the sum of the forest distance (instead of the effective resistance) of all ordered vertex pairs $(u,v)$:
\begin{equation}\label{eqn:total_forest_distance}
    R_f(G) := \sum_{u < v} \fd{G}{u}{v}.
\end{equation}
With an argument analogous to the one for total effective resistance, the forest index can be expressed using the trace as well:
\begin{equation}\label{eqn:forest_index_trace}
    R_f(G) = n \cdot \trace{\Omega} - n.
\end{equation}

\paragraph*{Total Harmonic Resistance (THR)}
We now propose a new measure to handle disconnected graphs, \emph{total harmonic resistance}. 
This measure is again based on the effective resistance; this time one sums up the reciprocal of all pairwise effective resistances -- therefore \emph{harmonic}:

\begin{equation}\label{eqn:total_harmonic_resistance}
    R_h(G) := \sum_{u < v} \frac{1}{\effres{G}{u}{v}}.
\end{equation}

For vertex pairs where the effective resistance is infinite (\iec vertices lie in different components), we define the reciprocal to be zero.
The reciprocity in this sum makes computations more difficult compared to the other two metrics.

\section{Related Work}
\label{sec:rel-work}
Due to its high relevance in numerous application areas as well as a rich assortment of research
questions, robustness in networks has been an active research area
for several decades~\cite{freitas2022graph}.
We thus point the interested reader to recent surveys for a broader overview~\cite{freitas2022graph,math9080895}.
Concerning robustness measures, the survey by Freitas \etal~\cite{freitas2022graph} categorizes them
into three classes: (i) based on structural (combinatorial) properties, (ii) spectral properties
of the adjacency matrix, and (iii) spectral properties of the Laplacian matrix.
Total effective resistance belongs to the third class as it can be computed by the sum of the
Laplacian (inverse) eigenvalues. Chan and Akoglu~\cite{chan2016optimizing} propose a budget-constrained
edge rewiring mechanism to address six different spectral measures -- a related optimization problem, 
yet different from \kgrodel.
Note that Oehlers and Fabian~\cite{math9080895} focus on communication
networks and use a more fine-grained categorization than Freitas \etal within their context.

Failures of components can result from various reasons, \egc from natural disasters, attacks, or wear.
The targeted attack models surveyed by Freitas \etal~\cite{freitas2022graph} refer to vertex
removals and are based on vertex degrees and centrality scores.
In general, vertex removals are the predominant failure model in the
literature; Newman~\cite[Ch.~15]{Newman2018networks} discusses percolation (removal of a fraction
of the nodes), for example. An important question in this context is after which fraction
the graph becomes disconnected or, more generally, when the giant component dissolves.
One can address this question analytically in generative models (\egc~\cite{Newman2018networks})
and/or empirically with real-world data (\egc~\cite{BEYGELZIMER2005593}).
As a prime example, an influential paper by Albert \etal~\cite{albert2000error} and follow-up work led to 
the popular belief that scale-free networks are ``robust-yet-fragile'', \iec robust against 
uniform vertex deletion and fragile against targeted attacks that remove high-degree vertices.
Recent work by Hasheminezhad and Brandes~\cite{DBLP:journals/ans/HasheminezhadB23} puts this view 
into a more nuanced perspective: robustness depends primarily on the graph's minimum degree, 
not a power-law degree distribution.

As mentioned in Sec.~\ref{sec:intro}, edge deletions are natural to model failures in numerous
applications. Liu \etal~\cite{liu24fast}, who study the problem of minimizing one node's information
centrality when removing $k$ edges, argue that edge deletions are less intrusive than vertex deletions
and that they provide a more fine-granular control of disruptions.
To measure how easy two vertices can reach each other via alternative paths, numerous works use effective
resistance~\cite{Ellens2011,ps18,liu24fast,Wang2014ImprovingRO,PredariBKM23}, whereas Zhu \etal~\cite{zhu2023measures} use forest distance~\cite{chebotarev2000forest} as summands of their
forest index, a metric related to effective resistance.
Both robustness measures express with lower values that more alternative pathways exist. A small total forest
distance (as well as a small total effective resistance) thus means that many
vertex pairs can reach each other via many alternative short paths.
Obviously, this is a desirable property for a robustness measure in a number of
applications, \egc when it comes to routing information or goods~\cite{math9080895}.
Forest distance has recently been used for forest closeness
centrality~\cite{JinBZ19forest,DBLP:conf/sdm/GrintenAPM21}. There and when used as part of the forest index,
it has the advantage (compared to the ordinary graph distance or effective resistance)
to be able to handle disconnected graphs without changes.

An exact solution of the \kgrodel problem with total harmonic resistance or a related measure 
is likely infeasible for instances of non-trivial size: (i) similar optimization problems have been shown to be 
$\mathcal{NP}$-hard~\cite{kooij2023minimizing,DBLP:conf/sdm/GrintenAPM21}, including the single-vertex
variant of Liu \etal~\cite{liu24fast} with information centrality, 
and (ii) mathematical programming, even when applied to related problems with simpler objective functions,
can usually solve instances with only hundreds or at most a few thousand vertices in reasonable time~\cite{DBLP:conf/alenex/AngrimanBDGGM21}.
Empirically, however, the related problem of adding $k$ edges to minimize
total effective resistance can be solved adequately (yet in general not optimally) with a standard
greedy algorithm~\cite{top15}.

we developed in our previous work~\cite{DBLP:conf/asunam/PredariKM22,PredariBKM23}

heuristics to accelerate the greedy algorithm for the $k$-GRIP problem 
(with usually tolerable losses in solution quality).
Even more closely related, Liu \etal~\cite{liu24fast} and Zhu \etal~\cite{zhu2023measures}
use greedy strategies to identify $k$ edges to delete from $G$ while optimizing for a robustness measure
(information centrality vs forest index).
We thus expect an adapted greedy algorithm to work similarly well for our variant of \kgrodel.

\section{Comparison of Exact Solutions}
\label{sec:exact-solutions}
To investigate the difference between forest index and total harmonic resistance as robustness measures for \kgrodel, we analyze exact solutions for a collection of small examples.
These examples consist of different graph classes: grid graphs and variants thereof, random graphs~\cite[Part~III]{Newman2018networks} generated using the Barab\'{a}si-Albert model (parameters: $k=3$, $n_{max}=18$), and random graphs generated with the Watts-Strogatz model (parameters: $n=16$, $\operatorname{deg}=3$, $p=0.7$).
For each example, we compute the exact solutions of the optimization problem for both robustness measures.
Results for the grid-like grahps are visualized in Figure~\ref{fig:gridSolutions}.
Due to symmetry in the grid-like graphs, there are often multiple optimal solutions; 
for simplicity, we only show one of these solutions.

\newlength{\staticSize}
\setlength{\staticSize}{2.55cm}
\begin{figure}[!hbt]
    \begin{center}
        \begin{tabular}{c|c}
            Forest Index & Total Harmonic Resistance \\[0.2cm]\hline\\[0.1cm]
            \resizebox{\staticSize}{!}{\input{grid5x3-FI.tex}} & \resizebox{\staticSize}{!}{\input{grid5x3-THR.tex}} \\[0.3cm]
            \resizebox{1.46\staticSize}{!}{\input{grid7x4-FI.tex}} & \resizebox{1.46\staticSize}{!}{\input{grid7x4-THR.tex}}       \\[0.3cm]
            \resizebox{1.24\staticSize}{!}{\input{grid5x6-FI.tex}} & \resizebox{1.24\staticSize}{!}{\input{grid5x6-THR.tex}} \\[0.3cm]
            \resizebox{1.68\staticSize}{!}{\input{hotdog5x6-FI.tex}} & \resizebox{1.68\staticSize}{!}{\input{hotdog5x6-THR.tex}} \\
        \end{tabular}
        \caption{Optimal solutions for $k=5$ on grid-like graphs using FI (left column) and THR (right column) as resistance measures. Edges highlighted in blue belong to the solution set.}
        \label{fig:gridSolutions}
    \end{center}
\end{figure}
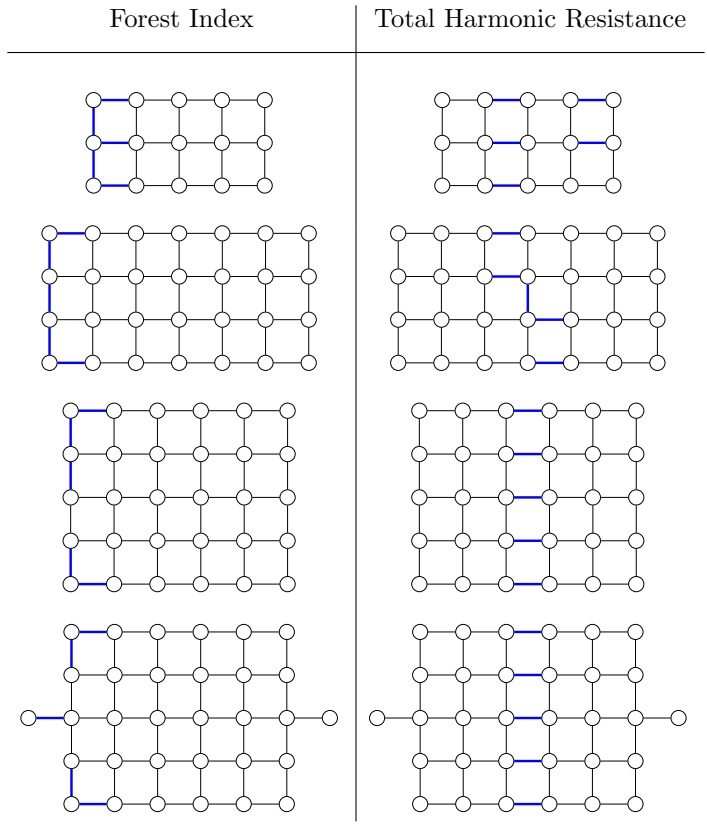

Visually, the figures suggest that the forest index FI finds edges in the periphery,
while THR finds central edges.
THR also seems to be more robust regarding changes to low-degree nodes in the periphery of a graph:
THR finds the same solution for the 4x7 grid and for the hotdog grid, while the FI solution changes.

To further support our claim about periphery vs center, we compute a centrality score 
for an edge set as follows:
Given the closeness centrality $c(\cdot)$
of all nodes in $G = (V,E)$ and a set of edges $S\subset E$, we rank the nodes by their closeness centrality and convert their rank into a relative (quantile) score $s\in[0,1]$, where the most central node has score 1 and the least central node has score 0.
Then, for each edge $e = (u,v)\in S$, we take the mean score of both incident nodes and call this the score of that edge $s(e)$.
Edges which are central in the graph have a larger score than less central edges.
Finally, we define the score of $S$ as the mean of all edge scores in the set.
\input{smallGraphsEdgeScores}
Scores for all solution sets are listed in Table~\ref{tab:smallGraphsSolutionScores}.
The centrality scores of the solutions further support our claim:
for all graphs of all three types, the score of the THR solution is higher (\iec more central)
than in the FI solution. This even holds when comparing the best FI solution to the worst THR 
solution in this metric.

\paragraph{Discussion.}
We would like to note that the observed behavior of the forest index is not according to our 
original intuition (which was similar to the one given by Zhu \etal~\cite{zhu2023measures})
\emph{before} working on this paper.
Broadly speaking, we generally expected the forest index to be maximized when the number of disconnected
vertex pairs is maximized, because this leads to many high terms in the sum.
The optimal solution (for appropriate $k$) on a grid graph would then be a balanced cut in the middle.
Instead, the optimal solution when using the forest index is a set of edges at the boundary of the grid, disconnecting just a few vertices from a large component.
While such peripheral edges may be desirable in some applications, we argue that in most scenarios
more central edges -- whose deletion ideally even leads to several connected components -- are beneficial
from an attacker's point of view. We further explore this in our case study on (parts of) the Berlin
road network in Section~\ref{sec:experiments}. 
To be able to process non-trivial instances, we propose to adapt a generic greedy algorithm in the
next section.

\section{Greedy Heuristic for \kgrodel}
\label{sec:heuristic}

We adapt the general greedy algorithm previously used for many related problems.
The basic idea of this algorithm is to iteratively pick the edge with best marginal loss until $k$ edges are found (see Algorithm~\ref{alg:greedy}).

The greedy algorithm starts by computing the Laplacian pseudoinverse, which is required to compute the effective resistance and hence the loss when removing an edge from $G$.
Then, in the main loop, iterate all edges in $G$ and compute the loss for each one.
Pick the best edge, update $G$ and $\Lpinv$, and repeat the main loop for $k$ iterations.

For the greedy algorithm, we need a formula to compute the marginal loss when removing an edge (Line~\ref{line:loss_evaluation}) as well as a way to update $\Lpinv$ after choosing an edge to compute the objective function in the next iteration (Line~\ref{line:update}).
These depend on the robustness metric used and will be derived in the next section.

For submodular functions the greedy framework can be combined with lazy evaluation~\cite{Minoux78} to speed up the computation. 
This lazy evaulation stores all candidates in a priority queue with their most recent loss value and instead of evaluating the loss for all candidates in each iteration of the main loop, it iteratively evaluates (and updates) only the top candidates' loss value until the top candidate is a candidate that has been evaluated in the current iteration of the main loop.
Effectively, this lazy evaluation reduces the number of evaluations significantly, while still providing a quality gurantee for submodular problems.
Even though \kgrodel is not known to be submodular for THR and is not submodular for FI~\cite{zhu2023measures}, we still apply this technique because practical experience has proven to lead to good results even for non-submodular problems~\cite{Minoux89}.

\begin{algorithm}[h]
    \begin{algorithmic}[1]
      \begin{small}
      \Function{Greedy}{$G$, $k$}
      \State \textbf{Input:} Graph $G=(V,E)$, $k \in \mathbb{N}$
      \State \textbf{Output:} $G_k$ -- graph after $k$ edge deletions
      \State $G_0 \gets G$
      \State Compute $\Lpinv$ 
      \For {$r \gets 0, \dots, k-1$}  \Comment main loop
            \State Compute $\text{loss}(e)\,\forall e\in E$   \Comment evaluation step \label{line:loss_evaluation}
            \State $e^* \gets \argmax_{e\in E} \text{loss}(e)$ 
            \State $G_{r+1} = G_r \setminus e^* $
            \State  \textsc{Update}($\Lpinv$, $G_{r+1}$)  \Comment {update step} \label{line:update}
      \EndFor
      \State \textbf{return} $G_{r+1}$
      \EndFunction
      \end{small}
    \end{algorithmic}
    \caption{Greedy algorithm for \kgrodel}
    \label{alg:greedy}
\end{algorithm}

Combining the lazy evaluation technique with the general greedy algorithm and THR-based loss and update functon leads to \GreedyTHR:
first, compute $\Lpinv_G$. This takes $\Oh(n^3)$ time (with standard tools in practice).
then, compute the loss for all edges of $G$ and set up a priority queue of all edges by their respective loss value.
In the main loop, get the top entry from the priority queue (using lazy evaluation), remove that edge from $G$ and update $\Lpinv_{G_r}$.

\subsection{Total Harmonic Resistance loss after Deleting an Edge}
\label{sub:harmonic-loss-update}

We now derive an update formula and state the loss formula for THR, which are required for the greedy algorithm.

For the \textsc{Update} step (Line~\ref{line:update}) there are efficient ways to compute $\Lpinv_{G'}$:
Removing an edge $e = \{a,b\}\in E$ from $G$ results in $G' = (V, E\setminus \{e\})$ and $\Lap_{G'} = \Lap_{G} - \easubeb\easubeb^T$, where $\vect{e}_i$ is the $i$-th unit vector.
One can apply the Sherman-Morrison-Formula~\cite{SMFormula} (which holds for $\Lap$ and $\easubeb$ as well) to write:
  \begin{align}
      \Lpinv_{G'} = \Lpinv_{G} + \frac{\Lpinv_{G}\easubeb\easubeb^T\Lpinv_{G}}{1-\easubeb^T\Lpinv_{G}\easubeb} \label{eqn:harmonic_lap_update_formula}\\
      =\Lpinv_{G} + \frac{\Lpinv_{G}\easubeb\easubeb^T\Lpinv_{G}}{1-\effres{G}{a}{b}}\nonumber
  \end{align}

There are limitations to using the Sherman-Morrison-Formula for updates: if the removed edge is a bridge, $\effres{G}{\cdot}{\cdot}=1$~\cite{Mavroforakis15} 
and hence the denominator in Eqn~\ref{eqn:harmonic_lap_update_formula} is 0.
In case the edge is not a bridge though, we can apply the Sherman-Morrison-Formula.

To handle the case of a bridge edge $e$, some more involved computation is required. 
Recall that $\Lap$ is a (permuted) block diagonal matrix where each block corresponds to a component of $G$ (see Section~\ref{sec:prelim}). 
Removing $e$ causes the corresponding block in $\Lap$ to be split into two blocks -- one for each component.
All other blocks of $\Lap$ are not modified by this edge removal.
Since the pseudoinverse of a block diagonal matrix is the block matrix build from the pseudoinverse of each block, $\Lpinv_{G'}$ can be found by computing the pseudoinverse of the two blocks related to $e$ and re-using the other pseudoinverse blocks from $\Lpinv_G$.

One has to keep track of a mapping from each node to its component (since in general $\Lap$ is permuted and we need to know which row/column belongs to which block) which takes $\Oh(n+m)$ time.
For simplicity, we re-compute this after removing a bridge edge (because the pseudoinversion step dominates the running time), but in princible it is possible to dynamically update the connected components instead.
The running time of the update step is either $\Oh(c^2)$ (non-bridge edge) or $\Oh(\max(n+m, c^3))$ (bridge edge), where $c$ is the size of the block matrix (resp. component of $G$) that contains $e$.

For the loss function, the basic formula is $\loss{a}{b} := \THR(G) - \THR(G') = \sum_{u < v} \frac{1}{\effres{G}{u}{v}} - \frac{1}{\effres{G'}{u}{v}}$. 
This formula depends on values in $\Lpinv_G$ and $\Lpinv_{G'}$ (via $\effres{G}{u}{v} = \Lpinv_G[u,u] + \Lpinv_G[v,v] - 2\Lpinv_G[u,v]$). 
Since this is a sum of reciprocal values, deriving an efficient formula proves difficult; we do not know of a closed formula analogous to the forest index or effective resistance yet.
Using the basic formula to compute the loss requires computing $\Lpinv_{G'}$ which we have discussed above. Computing the loss when $\Lpinv_{G'}$ is given takes $\Oh(n^2)$ time and the loss is computed up to $\Oh(km)$ times in the greedy algorithm (in the worst case, the loss is computed for each edge even though we use lazy evaluation). This leads to $\Oh(kmn^2\cdot\max(n+m, c^3))$ time overall for the loss computation. The running time varies a lot depending on the size and number of the components in $G$ and the number of bridge edges.

\subsection{Forest Index loss after Deleting an Edge}
\label{sub:forest-loss-update}

For our experiments, we are also interested in results from the greedy algorithm for FI.
Since the experimental setup by Zhu \etal~\cite{zhu2023measures} is to our knowledge not publicly available, we implemented our own version of this algorithm.
There are two differences in our implementation compared to the algorithm description given in their paper though:
(i) we exploit a connection between forest index and effective resistance to convert the FI computation back to a problem that is based on the Laplacian matrix.
This allows re-use of specialized Laplacian pseudoinverse solvers.
(ii) we use the lazy evaluation technique described in Section~\ref{sec:heuristic}, even though FI is not submodular. 
As mentioned, this technique usually yields good results even for non-submodular problems and in our preliminary experiments we observed no difference in the solution quality.

To derive a marginal loss formula for the forest index, we use a theorem on the connection between effective resistance
and forest distance; it allows to reduce the forest index formula (based on forest distance) 
back to total effective resistance and this reduction facilitates the reuse of some other theorems and algorithms:
  \begin{theorem}\label{thm:forest-distance-resistance}
      Given $G=(V,E)$, define the \emph{augmented Graph} $G_* = (V_*, E_*)$ with a universal vertex $\fstar$ which is connected to all other vertices: $V_* = V \cup \{\fstar\}$ and $E_* = E \cup \{(v, \fstar) : v\in V\}$.
  
      Then $\fd{G}{u}{v} = \effres{G_*}{u}{v} ~ \forall u,v\in V$.
  \end{theorem}
  The proof (with a slight change in the forest distance definition that does not affect the validity of the result)
  can be found in Ref.~\cite[Proposition~7]{chebotarev2000forest}.
   From Thm.~\ref{thm:forest-distance-resistance} the following result can be derived:
  
  \begin{proposition}\label{cor:totals_equal}
      The forest index can be written in terms of the effective resistance of the augmented graph: 
      $R_f(G) =  n\cdot\trace{\Lpinv_{G_*}} - (n+1)\cdot\Lpinv_{G_*}[\fstar,\fstar]$.
  \end{proposition}
  \begin{proof}
    See Appendix~\ref{sub:proof-loss}. The main idea is to extend the forest index sum by adding a zero term, which then includes the trace of $\Lpinv_{G_*}$.
  \end{proof}

  \paragraph{Edge Removal.}
  We can now use Proposition~\ref{cor:totals_equal} to write the forest index $R_f(G)$ in terms of the augmented graph $G_*$ and $\Lpinv_{G_*}$, which allows us to compute the marginal loss via $\Lpinv_{G_*}$ when removing an edge from $G$.
  
  Removing an edge $\{a,b\}\in E$ from $G$ results in $G' = (V, E\setminus \{\{a,b\}\})$ and $\Lap_{G'_*} = \Lap_{G_*} - \easubeb\easubeb^T$, where $\vect{e}_i$ is the $i$-th unit vector.
  Apply the Sherman-Morrison-Formula~\cite{SMFormula} to write:
  \begin{align}
      \Lpinv_{G'_*} =\Lpinv_{G_*} + \frac{\Lpinv_{G_*}\easubeb\easubeb^T\Lpinv_{G_*}}{1-\effres{G_*}{a}{b}}\label{eqn:forest_lap_update_formula}
  \end{align}

  To calculate the $\loss{a}{b} := R_f(G') - R_f(G)$ when removing $e=\{a,b\}$ from $G$, we can use Eqs.~(\ref{eqn:effective_resistance_trace}) and~(\ref{eqn:forest_lap_update_formula}) and the connection to total effective resistance (Proposition~\ref{cor:totals_equal}):
  
  \begin{proposition}
    \label{prop:loss}
      The marginal loss for the forest index when removing edge $(a,b)$ from $G$ is:
            
  \begin{align*}
    \loss{a}{b}
    = \frac{n}{1-\effres{G_*}{a}{b}}\cdot\norm{\Lpinv_{G_*}[:,a] - \Lpinv_{G_*}[:,b]}^2  \\
    - \frac{n+1}{1-\effres{G_*}{a}{b}}\cdot(\Lpinv_{G_*}[\fstar,a] - \Lpinv_{G_*}[\fstar,b])^2.
\end{align*}
  \end{proposition}
  \begin{proof}
    See Appendix~\ref{proof:lemma-loss}.
  \end{proof}

We use these loss and update formulae in our greedy algorithm, which allows us to re-use existing code.
Running times for the loss and update computation are $\Oh(n)$ and $\Oh(n^2)$ respectively, which results in a overall running time of $\Oh(kmn)$ and $\Oh(kn^2)$ for $k$ iterations of the greedy algorithm.

\section{Experimental Results}
\label{sec:experiments}

\subsection{Experimental Setup}

We conduct experiments to evaluate the quality of the greedy solution for THR to the greedy solution for FI (\GreedyTHR and \GreedyFI).
Our algorithms are implemented in C++ using the NetworKit toolkit~\cite{DBLP:books/sp/22/AngrimanGHMP22} as a graph library. We also build upon the previous work in 
Refs.~\cite{DBLP:conf/asunam/PredariKM22,PredariBKM23}.
To solve linear systems and compute the pseudoinverse, we use the LAMG solver from NetworKit.
\tool{SimExPal}~\cite{DBLP:journals/algorithms/AngrimanGLMNPT19} is used to manage our experiments and analyze the results.
All experiments are run on a machine with an Intel Xeon 6126 CPU and 192 GB RAM.

Code and the experimental setup are available on github: \url{https://github.com/bernlu/GRoDel-THR-FI}.

{
Table~\ref{tab:graphs} lists all networks used in our case study and benchmark study with their approximate number of nodes and edges.
For the following analysis, we split them into two groups: \emph{small} graphs with $|V| < 50K$ and \emph{large} graphs with $|V| > 50K$.
These networks are taken from SNAP~\cite{snap}, Networkrepository~\cite{NR} and KONECT~\cite{DBLP:conf/www/Kunegis13}.
For our experiments, we perform preprocessing on these graphs to turn them into simple graphs by removing self-loops, multi-edges and edge weights; we use the largest connected component of each graph.
We set the accuracy parameter $\epsilon$ of our LAMG solver (which we use to compute $\Lpinv$) to $10^{-5}$.

\begin{table}
    \centering
    \begin{small}
    \caption{Graph instances used for experiments, their vertex and edge counts after preprocessing, and the mean closeness centrality of the greedy THR and FI solutions.}\label{tab:graphs}
    \begin{tabular}{|lrr|cc|}
        \hline
        Graph                 & $|V|$ & $|E|$ & THR & FI \\ \hline                                           
        euro-road              & 1K   & 1.3K  &  0.661 & 	0.160 \\
        EmailUniv             & 1K    & 5.4K  & 0.405    & 0.526  \\   
        air-traffic-control & 1.2K & 2.4K &  0.579 &  0.063 \\ 
        inf-power             & 4K    & 6K  & 0.858    & 0.257  \\
        web-spam              & 4K    & 37K  & 0.457   & 0.039 \\
        Bcspwr10              & 5.3K  & 8.2K & 0.301  & 0.740 \\ 
        Erdos992              & 6K    & 7.5K & 0.643&0.691 \\                                        
        Reality               & 6.8K  & 7.6K & 0.825 &0.508 \\   
        \hline                                        
        Mitte-Berlin-Germany  & 1K    & 1.5K & 0.648 & 0.334 \\   
        Treptow-Köpenick-Berlin-Germany & 3.6K & 5.2K  &0.733 & 0.283\\
        \hline
    \end{tabular}
    \end{small}
\end{table}

\subsection{Case Study: Berlin Districts}

For our case study, we use the road networks of two Berlin districts (Table~\ref{tab:graphs}).
We choose these networks because road networks in general are easy to visualize and understand intuitively; Berlin specifically has some rivers flowing through the city which create cuts for many districts.
These river bridges make for a natural solution to \kgrodel which we will use as a manually chosen solution to compare to the greedy solutions.

Our graphs are generated from OpenStreetMap~\cite{OpenStreetMap} data using the \tool{osmnx}~\cite{boeing2017osmnx} python library.
We convert the data into a simple graph and use our NetworKit-based greedy algorithm to find the solution for both THR and FI.

}
The solutions for the \emph{Mitte} district are drawn on a map for visual inspection (Figure~\ref{fig:berlin_mitte}).
One can clearly see that \GreedyTHR finds some of the river bridges and other main roads, while \GreedyFI finds less important streets.
We also compare the solutions to the hand-picked solution that consists of the seven river bridges in this district (12 edges in total in our network because of multi-lane bridges).
The THR of this manual solution is larger than that of the \GreedyTHR solution, even though the greedy solution was computed for $k=20$ edges -- this further indicates that THR is a metric that prioritizes edges in a way we consider desirable.
In contrast to the previous observation, the FI score of the manual solution is \emph{worse} than the solution (of the same size) found by \GreedyFI; this is another hint that FI prioritizes edges in the periphery of a network.
We have also investigated other districts of Berlin, with very similar results: \GreedyTHR finds some bridges and large streets; \GreedyFI finds edges in the periphery and a manual choice of river bridges is better than the greedy solution.

\begin{figure}[tb]%
    \centering
        \includegraphics[width=\textwidth]{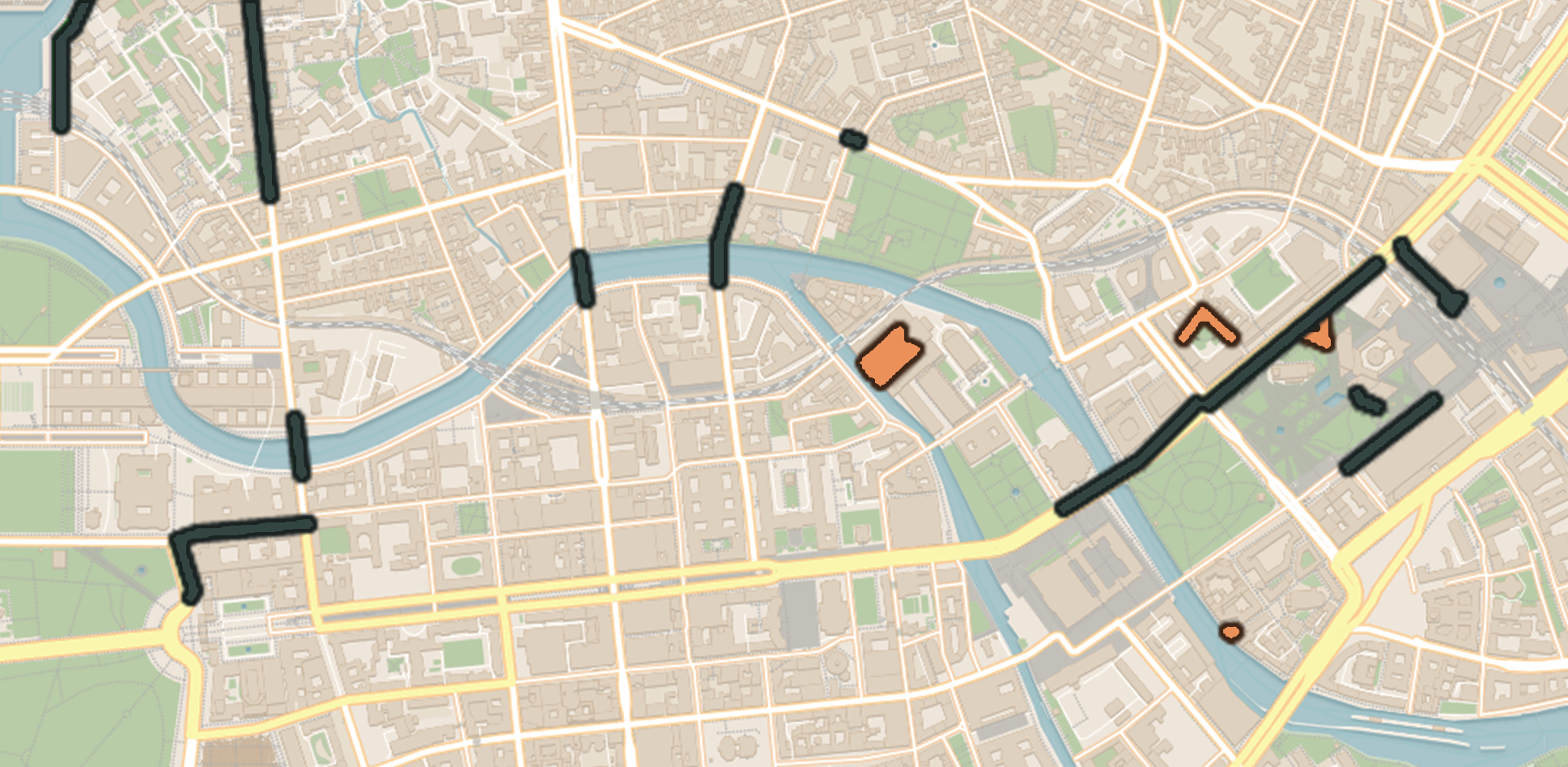}
    \caption{Berlin case study result. Grey edges are the \GreedyTHR solution; orange edges are the \GreedyFI solution. The image is cropped -- not all edges in the solution are displayed. \GreedyTHR finds four of seven river bridges while \GreedyFI mostly finds residential roads. Image created using OpenStreetMap~\cite{OpenStreetMap}}\label{fig:berlin_mitte}
\end{figure}

\subsection{Benchmark Results}
For the benchmark graphs, we evaluate the results by comparing the average closeness centrality of the solutions using the same we method described for the exact solutions in Section~\ref{sec:exact-solutions}.

Results are available in Table~\ref{tab:graphs}.
For most benchmark graphs we observe that the \GreedyTHR solution is considerably more central than the \GreedyFI solution; on average, the \GreedyTHR solution is about 25\% more central in the closeness centrality metric.
In the \emph{Bcspwr10} graph \GreedyTHR provides a considerably less central solution than \GreedyFI{} -- though there is no obvious reason for this result.

Regarding running times, with a timeout of 12 hours we found solutions for graphs with up to $6.8$K nodes or $13$K edges.
We observe that the network structure (esp. the amount of bridge edges) has significant impact on running times -- which is expected given the two ways to compute the update step, where the update for bridge edges is much more expensive.
As expected, running times for \GreedyFI are 2-4 orders of magnitude lower than \GreedyTHR.
The reason for this is that we can use the much more efficient loss formula using the trace of $\Lpinv$ for \GreedyFI while we do not know of an analogous formula for \GreedyTHR.

\section{Conclusions}
\label{sec:conclusions}
With the protection of large infrastructure in mind, we considered the \kgrodel problem to identify 
a set of $k$ particularly vulnerable edges in a graph. To this end, we proposed total harmonic
resistance as objective function and compared it against the recently proposed forest index.

We show with small examples where we compute the exact solution that total harmonic resistance prioritizes more central edges than the forest index.
We adapt the general greedy algorithm for similar optimization problems to \kgrodel with total harmonic resistance and show in a case study on the Berlin road network that THR favors more central edges in larger examples as well.
Finally, we run benchmark experiments which show that THR mostly favors more central edges than FI in a range of different network types.
We note that the greedy algorithm for THR has higher time complexity than the greedy algorithm for FI and our experiments confirm this in practice.

In the future, we would like to focus on speeding up the greedy algorithm for THR by improving the update and loss formulae and by finding other, faster heuristics. These are highly complex problems because of the reciprocity in the objective function -- which prevents re-use of most of the results and techniques used for related robustness measures like total effective resistance or forest index.

\paragraph{Acknowledgments}

We would like to thank Rob Kooij from TU Delft for insightful discussions on total harmonic resistance and many related measures.

\bibliography{references}

\pagebreak
\appendix
\setcounter{table}{0}
\renewcommand{\thetable}{A\arabic{table}}
\section{Appendix}

\subsection{Proof of Proposition~\ref{cor:totals_equal}}
\label{sub:proof-loss}
\begin{proof}
    Let $G=(V,E), u,v\in V$ and $G_* = (V_*, E_*)$ the augmented graph as described in Theorem~\ref{thm:forest-distance-resistance} with universal vertex $\fstar$. Then:
    \begin{align*}
        &R_f(G) = \sum_{u < v}  \fd{G}{u}{v} 
        = \sum_{u < v} \effres{G_*}{u}{v} \\
        &= \frac{1}{2} \sum_{u,v\in V}  \effres{G_*}{u}{v}  \\ 
        &= \frac{1}{2} \sum_{u\in V}  \sum_{v\in V} (\Lpinv_{G_*}[u,u] + \Lpinv_{G_*}[v,v] - 2~\Lpinv_{G_*}[u,v]) \\
        &= \frac{1}{2} \sum_{u\in V}  \sum_{v\in V}\Lpinv_{G_*}[u,u] +  \frac{1}{2}\sum_{u\in V}  \sum_{v\in V} \Lpinv_{G_*}[v,v] \eqnlinebreak - \frac{1}{2}\sum_{u\in V}  \sum_{v\in V}2~\Lpinv_{G_*}[u,v] \\
        &= \frac{n}{2} \sum_{u\in V}  \Lpinv_{G_*}[u,u] +  \frac{n}{2}  \sum_{v\in V} \Lpinv_{G_*}[v,v] \eqnlinebreak - \sum_{u\in V}  \sum_{v\in V}\Lpinv_{G_*}[u,v] \\
        &= n \sum_{u\in V}  \Lpinv_{G_*}[u,u] - \sum_{u\in V}  \sum_{v\in V}\Lpinv_{G_*}[u,v] \\
        &= n (\trace{\Lpinv_{G_*}} - \Lpinv_{G_*}[\fstar,\fstar]) - \sum_{u\in V}  -\Lpinv_{G_*}[u,\fstar] \\
        &= n\cdot\trace{\Lpinv_{G_*}} - n\cdot\Lpinv_{G_*}[\fstar,\fstar] + \sum_{u\in V}  \Lpinv_{G_*}[u,\fstar]\\
        &= n\cdot\trace{\Lpinv_{G_*}} - n\cdot\Lpinv_{G_*}[\fstar,\fstar] - \Lpinv_{G_*}[\fstar,\fstar] \\
        &= n\cdot\trace{\Lpinv_{G_*}} - (n+1)\cdot\Lpinv_{G_*}[\fstar,\fstar]\\
    \end{align*}
\end{proof}

\pagebreak

\subsection{Proof of Proposition~\ref{prop:loss}}\label{proof:lemma-loss}
\begin{proof}
    
\begin{align*}
    &\loss{a}{b}  = R_f(G') - R_f(G) \\
    &= n\cdot\trace{\Lpinv_{G_*'}} - (n+1)\cdot\Lpinv_{G_*'}[\fstar,\fstar] \eqnlinebreak  - (n\cdot\trace{\Lpinv_{G_*}} - (n+1)\cdot\Lpinv_{G_*}[\fstar,\fstar])   \\
    &= n\cdot\trace{\Lpinv_{G_*'}} - n\cdot\trace{\Lpinv_{G_*}} \eqnlinebreak  - (n+1)\cdot\Lpinv_{G_*'}[\fstar,\fstar]  + (n+1)\cdot\Lpinv_{G_*}[\fstar,\fstar]   \\
    &= n\cdot(\trace{\Lpinv_{G_*'}} - \trace{\Lpinv_{G_*}}) \eqnlinebreak  + (n+1)\cdot(\Lpinv_{G_*}[\fstar,\fstar]-  \Lpinv_{G_*'}[\fstar,\fstar])  \\
    &= n\cdot\trace{\Lpinv_{G_*'} - \Lpinv_{G_*}} \eqnlinebreak  + (n+1)\cdot(\Lpinv_{G_*}[\fstar,\fstar]- \Lpinv_{G_*'}[\fstar,\fstar])  \\
    &= n\cdot\trace{\Lpinv_{G_*} + \frac{\Lpinv_{G_*}\easubeb\easubeb^T\Lpinv_{G_*}}{1-\effres{G_*}{a}{b}} - \Lpinv_{G_*}} \eqnlinebreak  + (n+1)\cdot(\Lpinv_{G_*}[\fstar,\fstar]- \Lpinv_{G_*'}[\fstar,\fstar]) \\
    &= \frac{n}{1-\effres{G_*}{a}{b}}\cdot\trace{\Lpinv_{G_*}\easubeb\easubeb^T\Lpinv_{G_*}} \eqnlinebreak  + (n+1)\cdot(\Lpinv_{G_*}[\fstar,\fstar]- \Lpinv_{G_*'}[\fstar,\fstar]) \\
    &= \frac{n}{1-\effres{G_*}{a}{b}}\cdot\norm{\Lpinv_{G_*}[:,a] - \Lpinv_{G_*}[:,b]}^2  \eqnlinebreak  + (n+1)\cdot(\Lpinv_{G_*}[\fstar,\fstar]- \Lpinv_{G_*'}[\fstar,\fstar]) \\
    &= \frac{n}{1-\effres{G_*}{a}{b}}\cdot\norm{\Lpinv_{G_*}[:,a] - \Lpinv_{G_*}[:,b]}^2  \eqnlinebreak  - \frac{n+1}{1-\effres{G_*}{a}{b}}\cdot(\Lpinv_{G_*}[\fstar,a] - \Lpinv_{G_*}[\fstar,b])^2
\end{align*}

The last line is derived as follows:
\begin{align*}
    &\Lpinv_{G_*}[\fstar,\fstar]- \Lpinv_{G_*'}[\fstar,\fstar] \\
    &= \uvec{\fstar}^T\Lpinv_{G_*}\uvec{\fstar}- \uvec{\fstar}^T\Lpinv_{G_*'}\uvec{\fstar} \\
    &= \uvec{\fstar}^T\Lpinv_{G_*}\uvec{\fstar}- \uvec{\fstar}^T(\Lpinv_{G_*} + \frac{\Lpinv_{G_*}\easubeb\easubeb^T\Lpinv_{G_*}}{1-\effres{G_*}{a}{b}})\uvec{\fstar} \\
    &= \uvec{\fstar}^T\Lpinv_{G_*}\uvec{\fstar}- \uvec{\fstar}^T\Lpinv_{G_*}\uvec{\fstar} \\ & - \uvec{\fstar}^T\frac{\Lpinv_{G_*}\easubeb\easubeb^T\Lpinv_{G_*}}{1-\effres{G_*}{a}{b}}\uvec{\fstar} \\
    &= - \uvec{\fstar}^T\frac{\Lpinv_{G_*}\easubeb\easubeb^T\Lpinv_{G_*}}{1-\effres{G_*}{a}{b}}\uvec{\fstar} \\
    &= \frac{-1}{1-\effres{G_*}{a}{b}}(\uvec{\fstar}^T\Lpinv_{G_*}\uvec{a} - \uvec{\fstar}^T\Lpinv_{G_*}\uvec{b})^2\\
    &= -\frac{(\Lpinv_{G_*}[\fstar,a] - \Lpinv_{G_*}[\fstar,b])^2}{1-\effres{G_*}{a}{b}}
\end{align*}
\end{proof}

\if #0

\begin{landscape}

\newcommand{\vtext}[1]{\rotatebox{270}{#1}}

\begin{table}[tb]
    \centering
    \caption{Absolute running time in seconds.}\label{tab:abs_runtimes}

    \resizebox{\pdfpagewidth}{!}{%
    \begin{tabular}{llrrrrrrrrrrrrrrrrrrr}
        \toprule
          & instance & \vtext{facebook-ego-combined} & \vtext{web-spam} & \vtext{inf-power} & \vtext{wiki-Vote} & \vtext{p2p-Gnutella09} & \vtext{p2p-Gnutella04} & \vtext{ca-HepPh} & \vtext{web-indochina-2004} & \vtext{web-webbase-2001} & \vtext{ca-AstroPh} & \vtext{as-caida20071105} & \vtext{cit-HepTh} & \vtext{ia-email-EU} & \vtext{loc-brightkite} & \vtext{soc-Slashdot0902} & \vtext{ia-wiki-Talk} & \vtext{livemocha} & \vtext{flickrEdges} & \vtext{road-usroads} \\
        k-size & Heuristic &  &  &  &  &  &  &  &  &  &  &  &  &  &  &  &  &  &  &  \\
        \midrule
        \multirow[t]{5}{*}{2} & stGreedy & 3.0 & 1.9 & 0.7 & 5.7 & 3.8 & 113.5 & 15.1 & 9.8 & 8.2 & 50.8 & 25.2 & 148.6 & 49.4 & 279.0 & 1084.5 & 984.9 & 7008.1 & 6399.9 & 989.3 \\
         & simplStoch & 3.6 & 2.2 & 0.7 & 6.2 & 4.0 & 110.8 & 15.7 & 10.7 & 8.2 & 52.4 & 28.0 & 154.3 & 46.5 & 260.1 & 1298.1 & 833.7 & 7114.4 & 6219.4 & 1035.7 \\
         & colStoch & NaN & NaN & 0.4 & 2.0 & 1.5 & 48.7 & NaN & 3.1 & 2.3 & NaN & 14.3 & NaN & 9.9 & 70.7 & 387.3 & 241.9 & NaN & NaN & 205.1 \\
         & simplStochJLT & 4.8 & 3.1 & 1.0 & 5.4 & 3.1 & 57.9 & 7.7 & 4.1 & 2.4 & 11.6 & 6.5 & 18.5 & 5.0 & 16.4 & 84.4 & 47.3 & 164.7 & 216.7 & 18.3 \\
         & colStochJLT & 3.5 & 3.4 & 0.7 & 4.4 & 2.4 & 48.8 & 6.0 & 3.0 & 2.0 & 8.8 & 3.3 & 15.4 & 30.8 & 14.7 & 87.1 & 45.6 & 152.4 & 151.2 & 25.9 \\
        \cline{1-21}
        \multirow[t]{5}{*}{5} & stGreedy & 3.0 & 1.9 & 0.7 & 5.8 & 4.3 & 73.0 & 15.0 & 10.6 & 9.5 & 51.4 & 30.3 & 149.1 & 53.1 & 293.3 & 1222.5 & 1325.9 & 7684.7 & 6758.0 & 1131.9 \\
         & simplStoch & 3.6 & 2.2 & 0.8 & 6.6 & 4.0 & 88.2 & 15.9 & 10.6 & 10.1 & 53.3 & 50.8 & 169.8 & 55.4 & 295.6 & 1499.5 & 842.3 & 7136.3 & 6256.4 & 1120.1 \\
         & colStoch & NaN & 1.2 & 0.6 & 2.7 & 2.1 & 584.7 & 6.3 & 4.4 & 3.6 & 21.1 & 9.9 & NaN & 17.5 & 112.7 & 463.3 & 289.6 & NaN & NaN & 275.0 \\
         & simplStochJLT & 8.7 & 6.6 & 1.5 & 10.7 & 6.0 & 684.5 & 14.2 & 7.6 & 4.3 & 20.0 & 10.5 & 33.4 & 32.3 & 30.1 & 439.3 & 76.1 & 303.6 & 283.1 & 28.8 \\
         & colStochJLT & 8.0 & 6.2 & 1.6 & 8.4 & 5.4 & 71.7 & 312.4 & 6.4 & 3.9 & 17.6 & 6.1 & 29.1 & 7.5 & 28.8 & 169.6 & 75.4 & 272.3 & 604.1 & 38.5 \\
        \cline{1-21}
        \multirow[t]{5}{*}{20} & stGreedy & 3.3 & 2.1 & 0.9 & 6.7 & 5.3 & 116.1 & 20.4 & 13.5 & 17.1 & 57.9 & 54.2 & 177.2 & 96.4 & 466.0 & 1427.1 & 1464.0 & 9787.9 & 7251.6 & 1953.3 \\
         & simplStoch & 3.6 & 2.4 & 1.0 & 7.3 & 5.4 & 88.2 & 21.4 & 13.4 & 24.5 & 60.8 & 99.1 & 216.6 & 111.2 & 469.3 & 1637.1 & 1349.1 & 7482.7 & 6924.4 & 1985.7 \\
         & colStoch & 3.9 & 2.7 & 1.6 & 5.5 & 4.5 & 96.1 & 12.3 & 9.3 & 9.4 & 30.0 & 30.1 & 100.3 & 58.8 & 174.9 & 839.3 & NaN & NaN & 3551.8 & 558.8 \\
         & simplStochJLT & 28.6 & 21.2 & 5.5 & 332.6 & 21.2 & 826.9 & 44.3 & 24.8 & 14.3 & 65.1 & 27.0 & 104.6 & 51.2 & 697.8 & 404.9 & 1476.5 & 956.9 & 952.8 & 88.8 \\
         & colStochJLT & 27.7 & 323.2 & 5.3 & 33.4 & 22.4 & 808.4 & 44.9 & 24.9 & 14.8 & 67.3 & 20.1 & 103.3 & 22.6 & 96.4 & 367.7 & 269.0 & 866.2 & 987.7 & 97.1 \\
        \cline{1-21}
        \multirow[t]{5}{*}{50} & stGreedy & 3.6 & 2.6 & 1.4 & 8.1 & 8.0 & 120.1 & 23.3 & 19.8 & 32.6 & 73.7 & 117.2 & 231.8 & 177.9 & 765.6 & 2553.4 & 2286.2 & 9215.1 & 8244.5 & 3526.4 \\
         & simplStoch & 3.9 & 2.9 & 1.5 & 8.9 & 8.3 & 80.6 & 24.2 & 18.6 & 29.2 & 72.0 & 178.3 & 317.7 & 223.7 & 815.1 & 2414.1 & 2378.8 & 8750.1 & 7408.3 & 3556.7 \\
         & colStoch & 6.6 & 5.1 & 3.5 & 9.5 & 8.3 & 726.1 & 23.1 & 16.6 & 19.7 & 53.5 & 51.2 & 151.9 & 141.2 & 304.5 & 1284.9 & NaN & 5339.5 & 4655.3 & 1026.7 \\
         & simplStochJLT & 76.4 & 56.2 & 12.6 & 86.6 & 55.0 & 1134.8 & 109.5 & 58.8 & 32.5 & 464.7 & 58.3 & 268.4 & 78.2 & 226.5 & 1148.6 & 1450.2 & 2062.2 & 2487.4 & 203.9 \\
         & colStochJLT & 75.8 & 54.5 & 13.9 & 81.8 & 55.1 & 1101.3 & 110.3 & 61.5 & 34.3 & 158.2 & 48.3 & 553.9 & 103.2 & 528.1 & 959.6 & 700.8 & 1991.7 & 2438.1 & 215.3 \\
        \cline{1-21}
        \multirow[t]{5}{*}{100} & stGreedy & 4.0 & 3.4 & 2.2 & 11.3 & 12.3 & 88.9 & 32.5 & 28.3 & 50.2 & 98.6 & 186.9 & 319.9 & 319.2 & 1308.6 & 3792.7 & 4032.7 & 11501.8 & 9702.1 & 6107.8 \\
         & simplStoch & 4.2 & 3.7 & 2.4 & 12.6 & 12.7 & 98.3 & 32.6 & 27.6 & 64.8 & 94.1 & 331.5 & 543.0 & 396.9 & 1396.6 & 3807.0 & 3930.4 & 12325.5 & 9452.5 & 6240.4 \\
         & colStoch & 12.8 & 8.3 & 5.8 & 15.4 & 13.6 & 888.1 & 37.7 & 26.6 & 31.7 & 84.1 & 105.4 & 218.3 & 234.6 & 449.1 & 2025.3 & 1392.7 & 8118.6 & 7152.6 & 1608.5 \\
         & simplStochJLT & 142.5 & 407.1 & 27.4 & 462.5 & 265.9 & 1788.5 & 216.3 & 121.4 & 68.7 & 625.9 & 136.2 & 804.2 & 178.5 & 749.5 & 2326.4 & NaN & 4284.6 & 4676.4 & 400.1 \\
         & colStochJLT & 445.3 & 108.2 & 26.3 & 456.0 & 101.8 & 3494.2 & 227.9 & 124.3 & 72.9 & 628.0 & 114.5 & 507.0 & 153.2 & 748.0 & 2367.1 & 1394.7 & 4455.4 & 4261.6 & 416.9 \\
        \cline{1-21}
        \bottomrule
        \end{tabular}      
}
\end{table}

\end{landscape}

\fi

\end{document}

%% file: grid5x3-FI.tex
\begin{tikzpicture}[scale=1]
    \tikzstyle{node}=[draw,circle,inner sep=1pt, minimum width=10pt]
    \tikzstyle{active}=[draw,circle,inner sep=1pt,ultra thick,red]
    \tikzstyle{legend}=[minimum width=2.5cm,minimum height=1cm]

    \node[node] (v00) at (0,0) {};
    \node[node] (v01) at (0,1) {};
    \node[node] (v02) at (0,2) {};
    \node[node] (v10) at (1,0) {};
    \node[node] (v11) at (1,1) {};
    \node[node] (v12) at (1,2) {};
    \node[node] (v20) at (2,0) {};
    \node[node] (v21) at (2,1) {};
    \node[node] (v22) at (2,2) {};
    \node[node] (v30) at (3,0) {};
    \node[node] (v31) at (3,1) {};
    \node[node] (v32) at (3,2) {};
    \node[node] (v40) at (4,0) {};
    \node[node] (v41) at (4,1) {};
    \node[node] (v42) at (4,2) {};

    \draw[ultra thick, blue] (v00) -- (v01);
    \draw[ultra thick, blue] (v00) -- (v10);

    \draw[ultra thick, blue] (v01) -- (v02);
    \draw[ultra thick, blue] (v01) -- (v11);

    \draw[ultra thick, blue] (v02) -- (v12);

    \draw[thin, black] (v10) -- (v20);
    \draw[thin, black] (v10) -- (v11);

    \draw[thin, black] (v11) -- (v12);
    \draw[thin, black] (v11) -- (v21);
    
    \draw[thin, black] (v12) -- (v22);

    \draw[thin, black] (v20) -- (v21);
    \draw[thin, black] (v20) -- (v30);

    \draw[thin, black] (v21) -- (v22);
    \draw[thin, black] (v21) -- (v31);

    \draw[thin, black] (v22) -- (v32);

    \draw[thin, black] (v30) -- (v31);
    \draw[thin, black] (v30) -- (v40);

    \draw[thin, black] (v31) -- (v32);
    \draw[thin, black] (v31) -- (v41);

    \draw[thin, black] (v32) -- (v42);

    \draw[thin, black] (v40) -- (v41);

    \draw[thin, black] (v41) -- (v42);

\end{tikzpicture}

%% file: grid5x3-THR.tex
\begin{tikzpicture}[scale=1]
    \tikzstyle{node}=[draw,circle,inner sep=1pt, minimum width=10pt]
    \tikzstyle{active}=[draw,circle,inner sep=1pt,ultra thick,red]
    \tikzstyle{legend}=[minimum width=2.5cm,minimum height=1cm]

    \node[node] (v00) at (0,0) {};
    \node[node] (v01) at (0,1) {};
    \node[node] (v02) at (0,2) {};
    \node[node] (v10) at (1,0) {};
    \node[node] (v11) at (1,1) {};
    \node[node] (v12) at (1,2) {};
    \node[node] (v20) at (2,0) {};
    \node[node] (v21) at (2,1) {};
    \node[node] (v22) at (2,2) {};
    \node[node] (v30) at (3,0) {};
    \node[node] (v31) at (3,1) {};
    \node[node] (v32) at (3,2) {};
    \node[node] (v40) at (4,0) {};
    \node[node] (v41) at (4,1) {};
    \node[node] (v42) at (4,2) {};

    \draw[thin, black] (v00) -- (v01);
    \draw[thin, black] (v00) -- (v10);

    \draw[thin, black] (v01) -- (v02);
    \draw[thin, black] (v01) -- (v11);

    \draw[thin, black] (v02) -- (v12);

    \draw[ultra thick, blue] (v10) -- (v20);
    \draw[thin, black] (v10) -- (v11);

    \draw[thin, black] (v11) -- (v12);
    \draw[ultra thick, blue] (v11) -- (v21);
    
    \draw[ultra thick, blue] (v12) -- (v22);

    \draw[thin, black] (v20) -- (v21);
    \draw[thin, black] (v20) -- (v30);

    \draw[thin, black] (v21) -- (v22);
    \draw[thin, black] (v21) -- (v31);

    \draw[thin, black] (v22) -- (v32);

    \draw[thin, black] (v30) -- (v31);
    \draw[thin, black] (v30) -- (v40);

    \draw[thin, black] (v31) -- (v32);
    \draw[ultra thick, blue] (v31) -- (v41);

    \draw[ultra thick, blue] (v32) -- (v42);

    \draw[thin, black] (v40) -- (v41);

    \draw[thin, black] (v41) -- (v42);

\end{tikzpicture}

%% file: grid7x4-FI.tex
\begin{tikzpicture}[scale=1]
    \tikzstyle{node}=[draw,circle,inner sep=1pt, minimum width=10pt]
    \tikzstyle{active}=[draw,circle,inner sep=1pt,ultra thick,red]
    \tikzstyle{legend}=[minimum width=2.5cm,minimum height=1cm]

    \node[node] (v00) at (0,0) {};
    \node[node] (v01) at (0,1) {};
    \node[node] (v02) at (0,2) {};
    \node[node] (v03) at (0,3) {};
    \node[node] (v10) at (1,0) {};
    \node[node] (v11) at (1,1) {};
    \node[node] (v12) at (1,2) {};
    \node[node] (v13) at (1,3) {};
    \node[node] (v20) at (2,0) {};
    \node[node] (v21) at (2,1) {};
    \node[node] (v22) at (2,2) {};
    \node[node] (v23) at (2,3) {};
    \node[node] (v30) at (3,0) {};
    \node[node] (v31) at (3,1) {};
    \node[node] (v32) at (3,2) {};
    \node[node] (v33) at (3,3) {};
    \node[node] (v40) at (4,0) {};
    \node[node] (v41) at (4,1) {};
    \node[node] (v42) at (4,2) {};
    \node[node] (v43) at (4,3) {};
    \node[node] (v50) at (5,0) {};
    \node[node] (v51) at (5,1) {};
    \node[node] (v52) at (5,2) {};
    \node[node] (v53) at (5,3) {};
    \node[node] (v60) at (6,0) {};
    \node[node] (v61) at (6,1) {};
    \node[node] (v62) at (6,2) {};
    \node[node] (v63) at (6,3) {};

    \draw[ultra thick, blue] (v00) -- (v01);
    \draw[ultra thick, blue] (v00) -- (v10);

    \draw[ultra thick, blue] (v01) -- (v02);
    \draw[thin, black] (v01) -- (v11);

    \draw[thin, black] (v02) -- (v12);
    \draw[ultra thick, blue] (v02) -- (v03);

    \draw[ultra thick, blue] (v03) -- (v13);

    \draw[thin, black] (v10) -- (v20);
    \draw[thin, black] (v10) -- (v11);

    \draw[thin, black] (v11) -- (v12);
    \draw[thin, black] (v11) -- (v21);
    
    \draw[thin, black] (v12) -- (v22);
    \draw[thin, black] (v12) -- (v13);

    \draw[thin, black] (v13) -- (v23);

    \draw[thin, black] (v20) -- (v30);
    \draw[thin, black] (v20) -- (v21);

    \draw[thin, black] (v21) -- (v22);
    \draw[thin, black] (v21) -- (v31);
    
    \draw[thin, black] (v22) -- (v32);
    \draw[thin, black] (v22) -- (v23);

    \draw[thin, black] (v23) -- (v33);

    \draw[thin, black] (v30) -- (v40);
    \draw[thin, black] (v30) -- (v31);

    \draw[thin, black] (v31) -- (v32);
    \draw[thin, black] (v31) -- (v41);
    
    \draw[thin, black] (v32) -- (v42);
    \draw[thin, black] (v32) -- (v33);

    \draw[thin, black] (v33) -- (v43);

    \draw[thin, black] (v40) -- (v41);
    \draw[thin, black] (v40) -- (v50);

    \draw[thin, black] (v41) -- (v42);
    \draw[thin, black] (v41) -- (v51);

    \draw[thin, black] (v42) -- (v52);
    \draw[thin, black] (v42) -- (v43);

    \draw[thin, black] (v43) -- (v53);

    \draw[thin, black] (v50) -- (v51);
    \draw[thin, black] (v50) -- (v60);

    \draw[thin, black] (v51) -- (v52);
    \draw[thin, black] (v51) -- (v61);

    \draw[thin, black] (v52) -- (v62);
    \draw[thin, black] (v52) -- (v53);

    \draw[thin, black] (v53) -- (v63);

    \draw[thin, black] (v60) -- (v61);

    \draw[thin, black] (v61) -- (v62);
    \draw[thin, black] (v62) -- (v63);

\end{tikzpicture}

%% file: grid7x4-THR.tex
\begin{tikzpicture}[scale=1]
    \tikzstyle{node}=[draw,circle,inner sep=1pt, minimum width=10pt]
    \tikzstyle{active}=[draw,circle,inner sep=1pt,ultra thick,red]
    \tikzstyle{legend}=[minimum width=2.5cm,minimum height=1cm]

    \node[node] (v00) at (0,0) {};
    \node[node] (v01) at (0,1) {};
    \node[node] (v02) at (0,2) {};
    \node[node] (v03) at (0,3) {};
    \node[node] (v10) at (1,0) {};
    \node[node] (v11) at (1,1) {};
    \node[node] (v12) at (1,2) {};
    \node[node] (v13) at (1,3) {};
    \node[node] (v20) at (2,0) {};
    \node[node] (v21) at (2,1) {};
    \node[node] (v22) at (2,2) {};
    \node[node] (v23) at (2,3) {};
    \node[node] (v30) at (3,0) {};
    \node[node] (v31) at (3,1) {};
    \node[node] (v32) at (3,2) {};
    \node[node] (v33) at (3,3) {};
    \node[node] (v40) at (4,0) {};
    \node[node] (v41) at (4,1) {};
    \node[node] (v42) at (4,2) {};
    \node[node] (v43) at (4,3) {};
    \node[node] (v50) at (5,0) {};
    \node[node] (v51) at (5,1) {};
    \node[node] (v52) at (5,2) {};
    \node[node] (v53) at (5,3) {};
    \node[node] (v60) at (6,0) {};
    \node[node] (v61) at (6,1) {};
    \node[node] (v62) at (6,2) {};
    \node[node] (v63) at (6,3) {};

    \draw[thin, black] (v00) -- (v01);
    \draw[thin, black] (v00) -- (v10);

    \draw[thin, black] (v01) -- (v02);
    \draw[thin, black] (v01) -- (v11);

    \draw[thin, black] (v02) -- (v12);
    \draw[thin, black] (v02) -- (v03);

    \draw[thin, black] (v03) -- (v13);

    \draw[thin, black] (v10) -- (v20);
    \draw[thin, black] (v10) -- (v11);

    \draw[thin, black] (v11) -- (v12);
    \draw[thin, black] (v11) -- (v21);
    
    \draw[thin, black] (v12) -- (v22);
    \draw[thin, black] (v12) -- (v13);

    \draw[thin, black] (v13) -- (v23);

    \draw[thin, black] (v20) -- (v30);
    \draw[thin, black] (v20) -- (v21);

    \draw[thin, black] (v21) -- (v22);
    \draw[thin, black] (v21) -- (v31);
    
    \draw[ultra thick, blue] (v22) -- (v32);
    \draw[thin, black] (v22) -- (v23);

    \draw[ultra thick, blue] (v23) -- (v33);

    \draw[ultra thick, blue] (v30) -- (v40);
    \draw[thin, black] (v30) -- (v31);

    \draw[ultra thick, blue] (v31) -- (v32);
    \draw[ultra thick, blue] (v31) -- (v41);
    
    \draw[thin, black] (v32) -- (v42);
    \draw[thin, black] (v32) -- (v33);

    \draw[thin, black] (v33) -- (v43);

    \draw[thin, black] (v40) -- (v41);
    \draw[thin, black] (v40) -- (v50);

    \draw[thin, black] (v41) -- (v42);
    \draw[thin, black] (v41) -- (v51);

    \draw[thin, black] (v42) -- (v52);
    \draw[thin, black] (v42) -- (v43);

    \draw[thin, black] (v43) -- (v53);

    \draw[thin, black] (v50) -- (v51);
    \draw[thin, black] (v50) -- (v60);

    \draw[thin, black] (v51) -- (v52);
    \draw[thin, black] (v51) -- (v61);

    \draw[thin, black] (v52) -- (v62);
    \draw[thin, black] (v52) -- (v53);

    \draw[thin, black] (v53) -- (v63);

    \draw[thin, black] (v60) -- (v61);

    \draw[thin, black] (v61) -- (v62);
    \draw[thin, black] (v62) -- (v63);

\end{tikzpicture}

%% file: grid5x6-FI.tex
\begin{tikzpicture}[scale=1]
    \tikzstyle{node}=[draw,circle,inner sep=1pt, minimum width=10pt]
    \tikzstyle{active}=[draw,circle,inner sep=1pt,ultra thick,red]
    \tikzstyle{legend}=[minimum width=2.5cm,minimum height=1cm]

    \node[node] (v00) at (0,0) {};
    \node[node] (v01) at (0,1) {};
    \node[node] (v02) at (0,2) {};
    \node[node] (v03) at (0,3) {};
    \node[node] (v04) at (0,4) {};
    \node[node] (v10) at (1,0) {};
    \node[node] (v11) at (1,1) {};
    \node[node] (v12) at (1,2) {};
    \node[node] (v13) at (1,3) {};
    \node[node] (v14) at (1,4) {};
    \node[node] (v20) at (2,0) {};
    \node[node] (v21) at (2,1) {};
    \node[node] (v22) at (2,2) {};
    \node[node] (v23) at (2,3) {};
    \node[node] (v24) at (2,4) {};
    \node[node] (v30) at (3,0) {};
    \node[node] (v31) at (3,1) {};
    \node[node] (v32) at (3,2) {};
    \node[node] (v33) at (3,3) {};
    \node[node] (v34) at (3,4) {};
    \node[node] (v40) at (4,0) {};
    \node[node] (v41) at (4,1) {};
    \node[node] (v42) at (4,2) {};
    \node[node] (v43) at (4,3) {};
    \node[node] (v44) at (4,4) {};
    \node[node] (v50) at (5,0) {};
    \node[node] (v51) at (5,1) {};
    \node[node] (v52) at (5,2) {};
    \node[node] (v53) at (5,3) {};
    \node[node] (v54) at (5,4) {};

    \draw[ultra thick, blue] (v00) -- (v01);
    \draw[ultra thick, blue] (v00) -- (v10);

    \draw[thin, black] (v01) -- (v02);
    \draw[thin, black] (v01) -- (v11);

    \draw[thin, black] (v02) -- (v12);
    \draw[ultra thick, blue] (v02) -- (v03);

    \draw[thin, black] (v03) -- (v13);
    \draw[ultra thick, blue] (v03) -- (v04);

    \draw[ultra thick, blue] (v04) -- (v14);

    \draw[thin, black] (v10) -- (v20);
    \draw[thin, black] (v10) -- (v11);

    \draw[thin, black] (v11) -- (v12);
    \draw[thin, black] (v11) -- (v21);
    
    \draw[thin, black] (v12) -- (v22);
    \draw[thin, black] (v12) -- (v13);

    \draw[thin, black] (v13) -- (v23);
    \draw[thin, black] (v13) -- (v14);
    
    \draw[thin, black] (v14) -- (v24);

    \draw[thin, black] (v20) -- (v30);
    \draw[thin, black] (v20) -- (v21);

    \draw[thin, black] (v21) -- (v22);
    \draw[thin, black] (v21) -- (v31);
    
    \draw[thin, black] (v22) -- (v32);
    \draw[thin, black] (v22) -- (v23);

    \draw[thin, black] (v23) -- (v33);
    \draw[thin, black] (v23) -- (v24);
    
    \draw[thin, black] (v24) -- (v34);

    \draw[thin, black] (v30) -- (v40);
    \draw[thin, black] (v30) -- (v31);

    \draw[thin, black] (v31) -- (v32);
    \draw[thin, black] (v31) -- (v41);
    
    \draw[thin, black] (v32) -- (v42);
    \draw[thin, black] (v32) -- (v33);

    \draw[thin, black] (v33) -- (v43);
    \draw[thin, black] (v33) -- (v34);
    
    \draw[thin, black] (v34) -- (v44);

    \draw[thin, black] (v40) -- (v41);
    \draw[thin, black] (v40) -- (v50);

    \draw[thin, black] (v41) -- (v42);
    \draw[thin, black] (v41) -- (v51);

    \draw[thin, black] (v42) -- (v52);
    \draw[thin, black] (v42) -- (v43);

    \draw[thin, black] (v43) -- (v53);
    \draw[thin, black] (v43) -- (v44);

    \draw[thin, black] (v44) -- (v54);

    \draw[thin, black] (v50) -- (v51);

    \draw[thin, black] (v51) -- (v52);
    \draw[thin, black] (v52) -- (v53);
    \draw[thin, black] (v53) -- (v54);

\end{tikzpicture}

%% file: grid5x6-THR.tex
\begin{tikzpicture}[scale=1]
    \tikzstyle{node}=[draw,circle,inner sep=1pt, minimum width=10pt]
    \tikzstyle{active}=[draw,circle,inner sep=1pt,ultra thick,red]
    \tikzstyle{legend}=[minimum width=2.5cm,minimum height=1cm]

    \node[node] (v00) at (0,0) {};
    \node[node] (v01) at (0,1) {};
    \node[node] (v02) at (0,2) {};
    \node[node] (v03) at (0,3) {};
    \node[node] (v04) at (0,4) {};
    \node[node] (v10) at (1,0) {};
    \node[node] (v11) at (1,1) {};
    \node[node] (v12) at (1,2) {};
    \node[node] (v13) at (1,3) {};
    \node[node] (v14) at (1,4) {};
    \node[node] (v20) at (2,0) {};
    \node[node] (v21) at (2,1) {};
    \node[node] (v22) at (2,2) {};
    \node[node] (v23) at (2,3) {};
    \node[node] (v24) at (2,4) {};
    \node[node] (v30) at (3,0) {};
    \node[node] (v31) at (3,1) {};
    \node[node] (v32) at (3,2) {};
    \node[node] (v33) at (3,3) {};
    \node[node] (v34) at (3,4) {};
    \node[node] (v40) at (4,0) {};
    \node[node] (v41) at (4,1) {};
    \node[node] (v42) at (4,2) {};
    \node[node] (v43) at (4,3) {};
    \node[node] (v44) at (4,4) {};
    \node[node] (v50) at (5,0) {};
    \node[node] (v51) at (5,1) {};
    \node[node] (v52) at (5,2) {};
    \node[node] (v53) at (5,3) {};
    \node[node] (v54) at (5,4) {};

    \draw[thin, black] (v00) -- (v01);
    \draw[thin, black] (v00) -- (v10);

    \draw[thin, black] (v01) -- (v02);
    \draw[thin, black] (v01) -- (v11);

    \draw[thin, black] (v02) -- (v12);
    \draw[thin, black] (v02) -- (v03);

    \draw[thin, black] (v03) -- (v13);
    \draw[thin, black] (v03) -- (v04);

    \draw[thin, black] (v04) -- (v14);

    \draw[thin, black] (v10) -- (v20);
    \draw[thin, black] (v10) -- (v11);

    \draw[thin, black] (v11) -- (v12);
    \draw[thin, black] (v11) -- (v21);
    
    \draw[thin, black] (v12) -- (v22);
    \draw[thin, black] (v12) -- (v13);

    \draw[thin, black] (v13) -- (v23);
    \draw[thin, black] (v13) -- (v14);
    
    \draw[thin, black] (v14) -- (v24);

    \draw[ultra thick, blue] (v20) -- (v30);
    \draw[thin, black] (v20) -- (v21);

    \draw[thin, black] (v21) -- (v22);
    \draw[ultra thick, blue] (v21) -- (v31);
    
    \draw[ultra thick, blue] (v22) -- (v32);
    \draw[thin, black] (v22) -- (v23);

    \draw[ultra thick, blue] (v23) -- (v33);
    \draw[thin, black] (v23) -- (v24);
    
    \draw[ultra thick, blue] (v24) -- (v34);

    \draw[thin, black] (v30) -- (v40);
    \draw[thin, black] (v30) -- (v31);

    \draw[thin, black] (v31) -- (v32);
    \draw[thin, black] (v31) -- (v41);
    
    \draw[thin, black] (v32) -- (v42);
    \draw[thin, black] (v32) -- (v33);

    \draw[thin, black] (v33) -- (v43);
    \draw[thin, black] (v33) -- (v34);
    
    \draw[thin, black] (v34) -- (v44);

    \draw[thin, black] (v40) -- (v41);
    \draw[thin, black] (v40) -- (v50);

    \draw[thin, black] (v41) -- (v42);
    \draw[thin, black] (v41) -- (v51);

    \draw[thin, black] (v42) -- (v52);
    \draw[thin, black] (v42) -- (v43);

    \draw[thin, black] (v43) -- (v53);
    \draw[thin, black] (v43) -- (v44);

    \draw[thin, black] (v44) -- (v54);

    \draw[thin, black] (v50) -- (v51);

    \draw[thin, black] (v51) -- (v52);
    \draw[thin, black] (v52) -- (v53);
    \draw[thin, black] (v53) -- (v54);

\end{tikzpicture}

%% file: hotdog5x6-FI.tex
\begin{tikzpicture}[scale=1]
    \tikzstyle{node}=[draw,circle,inner sep=1pt, minimum width=10pt]
    \tikzstyle{active}=[draw,circle,inner sep=1pt,ultra thick,red]
    \tikzstyle{legend}=[minimum width=2.5cm,minimum height=1cm]

    \node[node] (vleft) at (-1,2) {};
    \node[node] (v00) at (0,0) {};
    \node[node] (v01) at (0,1) {};
    \node[node] (v02) at (0,2) {};
    \node[node] (v03) at (0,3) {};
    \node[node] (v04) at (0,4) {};
    \node[node] (v10) at (1,0) {};
    \node[node] (v11) at (1,1) {};
    \node[node] (v12) at (1,2) {};
    \node[node] (v13) at (1,3) {};
    \node[node] (v14) at (1,4) {};
    \node[node] (v20) at (2,0) {};
    \node[node] (v21) at (2,1) {};
    \node[node] (v22) at (2,2) {};
    \node[node] (v23) at (2,3) {};
    \node[node] (v24) at (2,4) {};
    \node[node] (v30) at (3,0) {};
    \node[node] (v31) at (3,1) {};
    \node[node] (v32) at (3,2) {};
    \node[node] (v33) at (3,3) {};
    \node[node] (v34) at (3,4) {};
    \node[node] (v40) at (4,0) {};
    \node[node] (v41) at (4,1) {};
    \node[node] (v42) at (4,2) {};
    \node[node] (v43) at (4,3) {};
    \node[node] (v44) at (4,4) {};
    \node[node] (v50) at (5,0) {};
    \node[node] (v51) at (5,1) {};
    \node[node] (v52) at (5,2) {};
    \node[node] (v53) at (5,3) {};
    \node[node] (v54) at (5,4) {};
    \node[node] (vright) at (6,2) {};

    \draw[ultra thick, blue] (v02) -- (vleft);

    \draw[ultra thick, blue] (v00) -- (v01);
    \draw[ultra thick, blue] (v00) -- (v10);

    \draw[thin, black] (v01) -- (v02);
    \draw[thin, black] (v01) -- (v11);

    \draw[thin, black] (v02) -- (v12);
    \draw[thin, black] (v02) -- (v03);

    \draw[thin, black] (v03) -- (v13);
    \draw[ultra thick, blue] (v03) -- (v04);

    \draw[ultra thick, blue] (v04) -- (v14);

    \draw[thin, black] (v10) -- (v20);
    \draw[thin, black] (v10) -- (v11);

    \draw[thin, black] (v11) -- (v12);
    \draw[thin, black] (v11) -- (v21);
    
    \draw[thin, black] (v12) -- (v22);
    \draw[thin, black] (v12) -- (v13);

    \draw[thin, black] (v13) -- (v23);
    \draw[thin, black] (v13) -- (v14);
    
    \draw[thin, black] (v14) -- (v24);

    \draw[thin, black] (v20) -- (v30);
    \draw[thin, black] (v20) -- (v21);

    \draw[thin, black] (v21) -- (v22);
    \draw[thin, black] (v21) -- (v31);
    
    \draw[thin, black] (v22) -- (v32);
    \draw[thin, black] (v22) -- (v23);

    \draw[thin, black] (v23) -- (v33);
    \draw[thin, black] (v23) -- (v24);
    
    \draw[thin, black] (v24) -- (v34);

    \draw[thin, black] (v30) -- (v40);
    \draw[thin, black] (v30) -- (v31);

    \draw[thin, black] (v31) -- (v32);
    \draw[thin, black] (v31) -- (v41);
    
    \draw[thin, black] (v32) -- (v42);
    \draw[thin, black] (v32) -- (v33);

    \draw[thin, black] (v33) -- (v43);
    \draw[thin, black] (v33) -- (v34);
    
    \draw[thin, black] (v34) -- (v44);

    \draw[thin, black] (v40) -- (v41);
    \draw[thin, black] (v40) -- (v50);

    \draw[thin, black] (v41) -- (v42);
    \draw[thin, black] (v41) -- (v51);

    \draw[thin, black] (v42) -- (v52);
    \draw[thin, black] (v42) -- (v43);

    \draw[thin, black] (v43) -- (v53);
    \draw[thin, black] (v43) -- (v44);

    \draw[thin, black] (v44) -- (v54);

    \draw[thin, black] (v50) -- (v51);

    \draw[thin, black] (v51) -- (v52);
    \draw[thin, black] (v52) -- (v53);
    \draw[thin, black] (v53) -- (v54);

    \draw[thin, black] (v52) -- (vright);

\end{tikzpicture}

%% file: hotdog5x6-THR.tex
\begin{tikzpicture}[scale=1]
    \tikzstyle{node}=[draw,circle,inner sep=1pt, minimum width=10pt]
    \tikzstyle{active}=[draw,circle,inner sep=1pt,ultra thick,red]
    \tikzstyle{legend}=[minimum width=2.5cm,minimum height=1cm]

    \node[node] (vleft) at (-1,2) {};
    \node[node] (v00) at (0,0) {};
    \node[node] (v01) at (0,1) {};
    \node[node] (v02) at (0,2) {};
    \node[node] (v03) at (0,3) {};
    \node[node] (v04) at (0,4) {};
    \node[node] (v10) at (1,0) {};
    \node[node] (v11) at (1,1) {};
    \node[node] (v12) at (1,2) {};
    \node[node] (v13) at (1,3) {};
    \node[node] (v14) at (1,4) {};
    \node[node] (v20) at (2,0) {};
    \node[node] (v21) at (2,1) {};
    \node[node] (v22) at (2,2) {};
    \node[node] (v23) at (2,3) {};
    \node[node] (v24) at (2,4) {};
    \node[node] (v30) at (3,0) {};
    \node[node] (v31) at (3,1) {};
    \node[node] (v32) at (3,2) {};
    \node[node] (v33) at (3,3) {};
    \node[node] (v34) at (3,4) {};
    \node[node] (v40) at (4,0) {};
    \node[node] (v41) at (4,1) {};
    \node[node] (v42) at (4,2) {};
    \node[node] (v43) at (4,3) {};
    \node[node] (v44) at (4,4) {};
    \node[node] (v50) at (5,0) {};
    \node[node] (v51) at (5,1) {};
    \node[node] (v52) at (5,2) {};
    \node[node] (v53) at (5,3) {};
    \node[node] (v54) at (5,4) {};
    \node[node] (vright) at (6,2) {};

    \draw[thin, black] (v02) -- (vleft);

    \draw[thin, black] (v00) -- (v01);
    \draw[thin, black] (v00) -- (v10);

    \draw[thin, black] (v01) -- (v02);
    \draw[thin, black] (v01) -- (v11);

    \draw[thin, black] (v02) -- (v12);
    \draw[thin, black] (v02) -- (v03);

    \draw[thin, black] (v03) -- (v13);
    \draw[thin, black] (v03) -- (v04);

    \draw[thin, black] (v04) -- (v14);

    \draw[thin, black] (v10) -- (v20);
    \draw[thin, black] (v10) -- (v11);

    \draw[thin, black] (v11) -- (v12);
    \draw[thin, black] (v11) -- (v21);
    
    \draw[thin, black] (v12) -- (v22);
    \draw[thin, black] (v12) -- (v13);

    \draw[thin, black] (v13) -- (v23);
    \draw[thin, black] (v13) -- (v14);
    
    \draw[thin, black] (v14) -- (v24);

    \draw[ultra thick, blue] (v20) -- (v30);
    \draw[thin, black] (v20) -- (v21);

    \draw[thin, black] (v21) -- (v22);
    \draw[ultra thick, blue] (v21) -- (v31);
    
    \draw[ultra thick, blue] (v22) -- (v32);
    \draw[thin, black] (v22) -- (v23);

    \draw[ultra thick, blue] (v23) -- (v33);
    \draw[thin, black] (v23) -- (v24);
    
    \draw[ultra thick, blue] (v24) -- (v34);

    \draw[thin, black] (v30) -- (v40);
    \draw[thin, black] (v30) -- (v31);

    \draw[thin, black] (v31) -- (v32);
    \draw[thin, black] (v31) -- (v41);
    
    \draw[thin, black] (v32) -- (v42);
    \draw[thin, black] (v32) -- (v33);

    \draw[thin, black] (v33) -- (v43);
    \draw[thin, black] (v33) -- (v34);
    
    \draw[thin, black] (v34) -- (v44);

    \draw[thin, black] (v40) -- (v41);
    \draw[thin, black] (v40) -- (v50);

    \draw[thin, black] (v41) -- (v42);
    \draw[thin, black] (v41) -- (v51);

    \draw[thin, black] (v42) -- (v52);
    \draw[thin, black] (v42) -- (v43);

    \draw[thin, black] (v43) -- (v53);
    \draw[thin, black] (v43) -- (v44);

    \draw[thin, black] (v44) -- (v54);

    \draw[thin, black] (v50) -- (v51);

    \draw[thin, black] (v51) -- (v52);
    \draw[thin, black] (v52) -- (v53);
    \draw[thin, black] (v53) -- (v54);

    \draw[thin, black] (v52) -- (vright);

\end{tikzpicture}

%% file: smallGraphsEdgeScores.tex
\begin{table}
    \caption{Solution set centrality scores as defined in Section~\ref{sec:exact-solutions}. Since multiple optimal solutions exist for some graphs, the score is computed for each solution and aggregated in this table.
}
    \label{tab:smallGraphsSolutionScores} 
    \centering




    \begin{tabular}{lrlrlrlrlrl}
        \toprule
        graph & \multicolumn{2}{c}{BA1} & \multicolumn{2}{c}{BA2} & \multicolumn{2}{c}{BA3} & \multicolumn{2}{c}{grid5x3} & \multicolumn{2}{c}{grid5x6} \\
        opt & FI & THR & FI & THR & FI & THR & FI & THR & FI & THR \\\midrule
        min & 0.31 & 0.40 & 0.29 & 0.47 & 0.32 & 0.40 & 0.24 & 0.53 & 0.09 & 0.69 \\
        mean & 0.33 & 0.40 & 0.29 & 0.47 & 0.34 & 0.40 & 0.24 & 0.60 & 0.10 & 0.69 \\
        max & 0.37 & 0.40 & 0.29 & 0.47 & 0.36 & 0.40 & 0.24 & 0.67 & 0.13 & 0.69 \\
        \bottomrule
        \end{tabular}
        \\[0.2em]
        \begin{tabular}{lrlrlrlrlrl}
            \toprule
            graph & \multicolumn{2}{c}{grid7x4} & \multicolumn{2}{c}{hotdog5x6} & \multicolumn{2}{c}{WS1} & \multicolumn{2}{c}{WS2} & \multicolumn{2}{c}{WS3} \\
            opt & FI & THR & FI & THR & FI & THR & FI & THR & FI & THR \\
            \midrule
            min & 0.11 & 0.76 & 0.14 & 0.71 & 0.34 & 0.49 & 0.27 & 0.34 & 0.16 & 0.28 \\
            mean & 0.11 & 0.76 & 0.14 & 0.71 & 0.34 & 0.49 & 0.27 & 0.34 & 0.16 & 0.28 \\
            max & 0.11 & 0.76 & 0.14 & 0.71 & 0.34 & 0.49 & 0.27 & 0.34 & 0.16 & 0.28 \\
            \bottomrule
            \end{tabular}

\end{table}

%% file: main.bbl
\begin{thebibliography}{10}
\providecommand{\url}[1]{\texttt{#1}}
\providecommand{\urlprefix}{URL }
\providecommand{\doi}[1]{https://doi.org/#1}

\bibitem{albert2000error}
Albert, R., Jeong, H., Barab{\'a}si, A.L.: Error and attack tolerance of
  complex networks. nature  \textbf{406}(6794),  378--382 (2000)

\bibitem{DBLP:conf/alenex/AngrimanBDGGM21}
Angriman, E., Becker, R., D'Angelo, G., Gilbert, H., van~der Grinten, A.,
  Meyerhenke, H.: Group-harmonic and group-closeness maximization -
  approximation and engineering. In: Proc. of the Symp. on Algorithm
  Engineering and Experiments, {ALENEX}. pp. 154--168. {SIAM} (2021).
  \doi{10.1137/1.9781611976472.12},
  \url{https://doi.org/10.1137/1.9781611976472.12}

\bibitem{DBLP:books/sp/22/AngrimanGHMP22}
Angriman, E., van~der Grinten, A., Hamann, M., Meyerhenke, H., Penschuck, M.:
  Algorithms for large-scale network analysis and the networkit toolkit. In:
  Bast, H., Korzen, C., Meyer, U., Penschuck, M. (eds.) Algorithms for Big Data
  - {DFG} Priority Program 1736, Lecture Notes in Computer Science, vol. 13201,
  pp. 3--20. Springer (2022). \doi{10.1007/978-3-031-21534-6\_1},
  \url{https://doi.org/10.1007/978-3-031-21534-6\_1}

\bibitem{DBLP:journals/algorithms/AngrimanGLMNPT19}
Angriman, E., van~der Grinten, A., von Looz, M., Meyerhenke, H.,
  N{\"{o}}llenburg, M., Predari, M., Tzovas, C.: Guidelines for experimental
  algorithmics: {A} case study in network analysis. Algorithms  \textbf{12}(7),
  ~127 (2019). \doi{10.3390/a12070127}, \url{https://doi.org/10.3390/a12070127}

\bibitem{barabasi}
Barab\'{a}si, A.L., P\'{o}sfai, M.: Network {S}cience. Cambridge University
  Press, Cambridge (2016)

\bibitem{BEYGELZIMER2005593}
Beygelzimer, A., Grinstein, G., Linsker, R., Rish, I.: Improving network
  robustness by edge modification. Physica A: Statistical Mechanics and its
  Applications  \textbf{357}(3),  593--612 (2005).
  \doi{https://doi.org/10.1016/j.physa.2005.03.040},
  \url{https://www.sciencedirect.com/science/article/pii/S0378437105003523}

\bibitem{boeing2017osmnx}
Boeing, G.: Osmnx: New methods for acquiring, constructing, analyzing, and
  visualizing complex street networks. Computers, environment and urban systems
   \textbf{65},  126--139 (2017)

\bibitem{DBLP:journals/socnet/BozzoF13}
Bozzo, E., Franceschet, M.: Resistance distance, closeness, and betweenness.
  Social Networks  \textbf{35}(3),  460--469 (2013).
  \doi{10.1016/j.socnet.2013.05.003},
  \url{https://doi.org/10.1016/j.socnet.2013.05.003}

\bibitem{oded}
Cats, O., Koppenol, G.J., Warnier, M.: Robustness assessment of link capacity
  reduction for complex networks: Application for public transport systems.
  Reliability Engineering \& System Safety  \textbf{167},  544--553 (2017)

\bibitem{chan2016optimizing}
Chan, H., Akoglu, L.: Optimizing network robustness by edge rewiring: a general
  framework. Data Mining and Knowledge Discovery  \textbf{30},  1395--1425
  (2016)

\bibitem{chebotarev2000forest}
Chebotarev, P.Y., Shamis, E.: The forest metrics of a graph and their
  properties. AUTOMATION AND REMOTE CONTROL C/C OF AVTOMATIKA I TELEMEKHANIKA
  \textbf{61}(8; ISSU 2),  1364--1373 (2000)

\bibitem{Ellens2011}
Ellens, W., Spieksma, F., {Van Mieghem}, P., Jamakovic, A., Kooij, R.:
  Effective graph resistance. Linear Algebra and its Applications
  \textbf{435}(10),  2491--2506 (2011)

\bibitem{freitas2022graph}
Freitas, S., Yang, D., Kumar, S., Tong, H., Chau, D.H.: Graph vulnerability and
  robustness: A survey. IEEE Transactions on Knowledge and Data Engineering
  (2022)

\bibitem{DBLP:conf/sdm/GrintenAPM21}
van~der Grinten, A., Angriman, E., Predari, M., Meyerhenke, H.: New
  approximation algorithms for forest closeness centrality - for individual
  vertices and vertex groups. In: Proceedings of the 2021 {SIAM} International
  Conference on Data Mining, {SDM} 2021. pp. 136--144. {SIAM} (2021).
  \doi{10.1137/1.9781611976700.16},
  \url{https://doi.org/10.1137/1.9781611976700.16}

\bibitem{DBLP:journals/ans/HasheminezhadB23}
Hasheminezhad, R., Brandes, U.: Robustness of preferential-attachment graphs.
  Appl. Netw. Sci.  \textbf{8}(1), ~32 (2023).
  \doi{10.1007/s41109-023-00556-5},
  \url{https://doi.org/10.1007/s41109-023-00556-5}

\bibitem{JinBZ19forest}
Jin, Y., Bao, Q., Zhang, Z.: Forest distance closeness centrality in
  disconnected graphs. In: 2019 IEEE International Conference on Data Mining
  (ICDM). pp. 339--348. IEEE Computer Society (nov 2019).
  \doi{10.1109/ICDM.2019.00044},
  \url{https://doi.ieeecomputersociety.org/10.1109/ICDM.2019.00044}

\bibitem{Klein93}
Klein, D., Randic, M.: Resistance distance. Journal of Mathematical Chemistry
  \textbf{12},  81--95 (12 1993). \doi{10.1007/BF01164627}

\bibitem{kooij2023minimizing}
Kooij, R.E., Achterberg, M.A.: Minimizing the effective graph resistance by
  adding links is np-hard. arXiv preprint arXiv:2302.12628  (2023)

\bibitem{yakup}
Koç, Y., Warnier, M., {Van Mieghem}, P., Kooij, R.E., Brazier, F.M.: A
  topological investigation of phase transitions of cascading failures in power
  grids. Physica A: Statistical Mechanics and its Applications  \textbf{415},
  273--284 (2014)

\bibitem{DBLP:conf/www/Kunegis13}
Kunegis, J.: {KONECT:} the koblenz network collection. In: Carr, L., Laender,
  A.H.F., L{\'{o}}scio, B.F., King, I., Fontoura, M., Vrandecic, D., Aroyo, L.,
  de~Oliveira, J.P.M., Lima, F., Wilde, E. (eds.) 22nd International World Wide
  Web Conference, {WWW} '13. pp. 1343--1350. International World Wide Web
  Conferences Steering Committee / {ACM} (2013). \doi{10.1145/2487788.2488173},
  \url{https://doi.org/10.1145/2487788.2488173}

\bibitem{snap}
Leskovec, J., Krevl, A.: {SNAP Datasets}: {Stanford} large network dataset
  collection. \url{http://snap.stanford.edu/data} (Jun 2014)

\bibitem{liu24fast}
Liu, C., Zhou, X., Zehmakan, A.N., Zhang, Z.: A fast algorithm for moderating
  critical nodes via edge removal. IEEE Transactions on Knowledge and Data
  Engineering  \textbf{36}(4),  1385--1398 (2024).
  \doi{10.1109/TKDE.2023.3309987}

\bibitem{Mavroforakis15}
Mavroforakis, C., Garcia-Lebron, R., Koutis, I., Terzi, E.: Spanning edge
  centrality: Large-scale computation and applications. In: Proc. of the 24th
  Intl. Conference on World Wide Web. p. 732–742. Intl. World Wide Web
  Conferences Steering Committee (2015)

\bibitem{Minoux78}
Minoux, M.: Accelerated greedy algorithms for maximizing submodular set
  functions. In: Stoer, J. (ed.) Optimization Techniques. pp. 234--243.
  Springer Berlin Heidelberg, Berlin, Heidelberg (1978)

\bibitem{Minoux89}
Minoux, M.: Networks synthesis and optimum network design problems: Models,
  solution methods and applications. Networks  \textbf{19}(3),  313--360
  (1989). \doi{10.1002/net.3230190305},
  \url{https://doi.org/10.1002/net.3230190305}

\bibitem{Newman2018networks}
Newman, M.: Networks. Oxford University Press, 2nd edn. (2018)

\bibitem{math9080895}
Oehlers, M., Fabian, B.: Graph metrics for network robustness—a survey.
  Mathematics  \textbf{9}(8) (2021). \doi{10.3390/math9080895},
  \url{https://www.mdpi.com/2227-7390/9/8/895}

\bibitem{OpenStreetMap}
{OpenStreetMap contributors}: {OpenStreetMap database }. \url{
  https://www.openstreetmap.org } (2017)

\bibitem{ps18}
Pizzuti, C., Socievole, A.: A genetic algorithm for enhancing the robustness of
  complex networks through link protection. In: International Conference on
  Complex Networks and their Applications. pp. 807--819. Springer (2018)

\bibitem{PredariBKM23}
Predari, M., Berner, L., Kooij, R., Meyerhenke, H.: Greedy optimization of
  resistance-based graph robustness with global and local edge insertions.
  Social Network Analysis and Mining  (2023), to appear. Also available as
  arXiv preprint 2309.08271

\bibitem{DBLP:conf/asunam/PredariKM22}
Predari, M., Kooij, R., Meyerhenke, H.: Faster greedy optimization of
  resistance-based graph robustness. In: {IEEE/ACM} International Conference on
  Advances in Social Networks Analysis and Mining, {ASONAM} 2022, Istanbul,
  Turkey, November 10-13, 2022. pp.~1--8. {IEEE} (2022)

\bibitem{NR}
Rossi, R.A., Ahmed, N.K.: The network data repository with interactive graph
  analytics and visualization. In: AAAI (2015),
  \url{http://networkrepository.com}

\bibitem{jose}
Rueda, D.F., Calle, E., Marzo, J.L.: Robustness comparison of 15 real
  telecommunication networks: Structural and centrality measurements. Journal
  of Network and Systems Management  \textbf{25}(2),  269--289 (Apr 2017)

\bibitem{SMFormula}
Sherman, J., Morrison, W.J.: {Adjustment of an Inverse Matrix Corresponding to
  a Change in One Element of a Given Matrix}. The Annals of Mathematical
  Statistics  \textbf{21}(1),  124 -- 127 (1950)

\bibitem{top15}
Summers, T., Shames, I., Lygeros, J., D{\"o}rfler, F.: Topology design for
  optimal network coherence. In: 2015 European Control Conference (ECC). pp.
  575--580. IEEE (2015)

\bibitem{Wang2014ImprovingRO}
Wang, X., Pournaras, E., Kooij, R.E., Mieghem, P.V.: Improving robustness of
  complex networks via the effective graph resistance. The European Physical
  Journal B  \textbf{87},  1--12 (2014)

\bibitem{water}
Yazdani, A., Jeffrey, P.: Complex network analysis of water distribution
  systems. Chaos (Woodbury, N.Y.)  \textbf{21},  016111 (03 2011)

\bibitem{zhu2023measures}
Zhu, L., Bao, Q., Zhang, Z.: Measures and optimization for robustness and
  vulnerability in disconnected networks. IEEE Transactions on Information
  Forensics and Security  (2023)

\end{thebibliography}
